%% file: main.tex
\documentclass[envcountsame, envcountsect, runningheads, a4paper]{llncs}

\usepackage[dvipsnames]{xcolor}
\usepackage{amsmath,amsfonts,amssymb}
\usepackage[utf8]{inputenc}
\usepackage{graphicx}
\usepackage{tikz}
\usepackage{hyperref}
\usepackage[backgroundcolor=magenta!10!white]{todonotes}
\usepackage{setspace}
\usepackage{enumitem}
\usepackage{ifthen}
\usepackage{caption}
\usepackage{microtype}
\usepackage{stackrel}
\usepackage{proof}
\usepackage[space]{cite}
\usepackage{numprint}
\usepackage{stmaryrd}
\usepackage{cleveref}
\usepackage{multirow}
\usepackage{subcaption}
\usepackage{verbatim}
\usepackage{placeins}

\usepackage{pgfplots}
\usepgfplotslibrary{fillbetween}
\pgfplotsset{width=3.5cm, compat=1.8, grid style={dashed}}

\input{tikzstyle}

\spnewtheorem{notation}[theorem]{Notation}{\bfseries}{\upshape}
\spnewtheorem{setting}[theorem]{Setting}{\bfseries}{\upshape}

\usepackage{thmtools}
\usepackage{relsize}
\usepackage{multirow}
\usepackage{multicol}

\usepackage{mathtools}
\DeclarePairedDelimiter\ceil{\lceil}{\rceil}

\usetikzlibrary{arrows,
	automata,
	calc,
	decorations,
	decorations.pathreplacing,
	positioning,
	shapes}

\setenumerate[1]{label={(\arabic*)}} 
\numberwithin{equation}{section}
\npdecimalsign{.}
\npthousandsep{,}

\makeatletter
\renewcommand\subsubsection{\@startsection{subsubsection}{3}{\z@}%
	{-6\p@ \@plus -6\p@ \@minus -6\p@}
	{-0.5em \@plus -0.22em \@minus -0.1em}%
	{\normalfont\normalsize\bfseries\boldmath}}
\makeatother

\newcommand{\R}{\mathbb{R}}
\renewcommand{\P}{\mathcal{P}}
\newcommand{\prb}{\mathbf{Pr}}
\newcommand{\Ab}{\mathbf{A}}

\newcommand{\Pb}{\mathbf{P}}

\newcommand{\yb}{\mathbf{y}}
\newcommand{\bb}{\mathbf{b}}

\newcommand{\M}{\mathcal{M}}
\newcommand{\Scal}{\mathcal{S}}

\newcommand{\vb}{\mathbf{v}}

\newcommand{\C}{\mathcal{C}}
\newcommand{\T}{\mathcal{T}}
\newcommand{\Val}{\operatorname{Val}}
\newcommand{\CC}{\operatorname{CC}}

\newcommand{\Loc}{\operatorname{Loc}}
\newcommand{\loc}{\operatorname{loc}}
\newcommand{\inv}{\operatorname{inv}}

\newcommand{\true}{\operatorname{true}}
\newcommand{\false}{\operatorname{false}}

\newcommand{\nb}{\boldsymbol{n}}
\newcommand{\eb}{\boldsymbol{e}}
\newcommand{\zb}{\mathbf{z}}
\newcommand{\all}{{\operatorname{all}}}
\newcommand{\Act}{\operatorname{Act}}
\newcommand{\supp}{\operatorname{supp}}
\newcommand{\last}{\operatorname{last}}
\newcommand{\Dist}{\operatorname{Dist}}
\newcommand{\goal}{\operatorname{goal}}
\newcommand{\fail}{\operatorname{fail}}

\renewcommand{\S}{\mathfrak{S}}
\renewcommand{\Pr}{\mathrm{Pr}}
\newcommand{\fin}{{\operatorname{fin}}}
\newcommand{\Paths}{\operatorname{Paths}}

\newcommand{\vol}{\operatorname{vol}}
\newcommand{\sem}{{\operatorname{sem}}}
\newcommand{\zerocl}{c_0}
\newcommand{\poly}{\operatorname{poly}}

\newcommand{\timeup}{\mathlarger{\uparrow}}
\newcommand{\hashP}{\texttt{\#P}}
\newcommand{\PP}{\texttt{PP}}
\newcommand{\PSPACE}{\texttt{PSPACE}}
\newcommand{\EXPTIME}{\texttt{EXPTIME}}

\newcommand{\pointsize}{0.5} 

\newcommand{\Implies}{\Vdash}

\newif\iflongversion
\longversiontrue

\iflongversion
\newcommand{\citeapp}[1]{\Cref{#1}}
\else
\newcommand{\citeapp}[1]{the technical report~\cite[\Cref{#1}]{}}
\fi

\iflongversion
\newcommand{\referapp}{the appendix}
\else
\newcommand{\referapp}{the technical report~\cite{}}
\fi

\begin{document}

\title{Minimal witnesses for \\probabilistic timed automata}
\author{Simon Jantsch \and Florian Funke \and Christel Baier}
\authorrunning{Simon Jantsch, Florian Funke, Christel Baier}
\titlerunning{Minimal witnesses for probabilistic timed automata}
%
\institute{Technische Universität Dresden\thanks{This work was funded by DFG grant 389792660 as part of \href{https://perspicuous-computing.science}{TRR~248}, the Cluster of Excellence EXC 2050/1 (CeTI, project ID 390696704, as part of Germany’s Excellence Strategy), DFG-projects BA-1679/11-1 and BA-1679/12-1, and the Research Training Group QuantLA (GRK 1763).}\\
\email{\{simon.jantsch, florian.funke, christel.baier\}@tu-dresden.de}}

\maketitle

\begin{abstract}
  Witnessing subsystems have proven to be a useful concept in the analysis of probabilistic systems, for example as diagnostic information on why a given property holds or as input to refinement algorithms.
  This paper introduces witnessing subsystems for reachability problems in probabilistic timed automata (PTA).
  Using a new operation on difference bounds matrices, it is shown how Farkas certificates of finite-state bisimulation quotients of a PTA can be translated into witnessing subsystems.
  We present algorithms for the computation of minimal witnessing subsystems under three notions of minimality, which capture the timed behavior from different perspectives, and discuss their complexity.
\end{abstract}

\section{Introduction}
\label{sec:intro}

A \emph{witnessing subsystem} is a part of a probabilistic system that by itself carries enough probability to satisfy a given constraint.
Hence, it provides insight into which components of the system are sufficient for the desired behavior, and on the other hand, which can be disabled without interfering with it.
The concept of witnessing subsystems (sometimes, dually, refered to as \emph{critical subsystems}) for discrete-time Markov chains (DTMC) and Markov decision processes (MDP) has received considerable attention\cite{JansenAKWKB11,JansenWAZKBS14, WimmerJAKB14, FunkeJB20}.
Apart from providing diagnostic information on why a property holds, witnessing subsystems have been used for automated refinement and synthesis algorithms\cite{HermannsWZ2008,CeskaHJK19}.

In this paper we introduce witnessing subsystems for reachability constraints in probabilistic timed automata (PTA) \cite{KwiatkowskaNSS02, Beauquier03}. PTAs combine real-time, non-deterministic, and probabilistic behavior and are a widely used formalism for the modeling and verification of reactive systems such as communication protocols and scheduler optimization tasks \cite{KwiatkowskaNS03, NormanPS13}. However, as the state space of PTAs is inherently uncountable, the theory of witnessing subsystems in finite-state probabilistic systems is not applicable. Our generalization applies to both maximal and minimal reachability probabilities, where particularly the latter needs to be treated with special care in the timed setting.

A continuous algebraic counterpart to witnessing subsystems in MDPs are \emph{Farkas certificates}, which are vectors certifying threshold properties of the form $\prb_{\M}^{\min}(\lozenge\goal)\geq\lambda$ or $\prb_{\M}^{\max}(\lozenge\goal)\geq\lambda$~\cite{FunkeJB20}. We pave a two-way street between witnessing subsystems in a PTA and Farkas certificates of finite-state bisimulation quotients by giving explicit procedures how one can be obtained from the other. It is noteworthy that this translation makes finite-state methods available for the certification of threshold properties in infinite-state models.

Relevant information from a subsystem can only be expected after optimization along suitable minimality criteria, the most prevalent of which for MDPs is state-minimality. In the timed setting, however, the usefulness of a minimality criterion is more volatile under changing the specific practical problem. For this reason, we introduce three notions of minimality aimed at finding witnessing subsystems with few locations, strong invariants, or small invariant volume.

In all three cases, we present single-exponential algorithms for the computation of minimal witnessing subsystems. They heavily rely on the connection between PTA subsystems and Farkas certificates of bisimulation quotients and can also be adapted to faster heuristic approaches.
Furthermore, we observe that while comparing two subsystems according to their location number or invariance strength is not difficult, it is inherently harder (\PP-hard) to compare their invariance volume. All omitted proofs can be found in~\referapp{}.

\subsubsection{Contributions.}
The notion of (strong) subsystem for PTAs is introduced (\Cref{def:subsystem}) and justified by proving that reachability probabilities do not increase under passage to a subsystem (\Cref{prop:monotonicity}).
It is shown that subsystems of a PTA induce Farkas certificates in time-abstracting bisimulation quotients (\Cref{prop:subsysandfarkascert}). Vice versa, a conceptual construction of PTA subsystems from Farkas certificates of such quotients is given (\Cref{def:indsubsys} and \Cref{prop:witsubsysandfarkcert}), which relies on a new operation on difference bounds matrices (DBMs), see \Cref{def:zone closure}.
 Three notions of minimality for PTA subsystems are introduced and compared. We present mixed integer linear programs for computing location- and invariance-minimal subsystems. Volume-minimal subsystems can be computed with the aid of a \emph{multi-objective} mixed integer linear program (\Cref{sec:minimal}).
 Regarding volume-minimality, we establish \PP-hardness of comparing two witnessing subsystems according to their volume (\Cref{prop:PP}).

\subsubsection{Related work.}
Exact and heuristic approaches for computing minimal and small witnessing subsystems in DTMCs have been proposed in~\cite{JansenAKWKB11,JansenWAZKBS14}, and generalizations to MDPs have been considered in~\cite{AndresDR08, WimmerJAKB14, FunkeJB20}.
The approach in~\cite{WimmerJAK15} is most closely related to our work as it finds counterexamples for a high-level description (a guarded command language for MDPs).
Model checking PTAs against PTCTL specifications has first been described in \cite{KwiatkowskaNSS02}. 
Subsequent approaches use digital clocks \cite{KwiatkowskaNPS07}, symbolic model checking techniques \cite{KwiatkowskaNSW07}, or the boundary region graph \cite{JurdzinskiKNT09}. The work \cite{BerendsenJK06} presents an algorithm for price-bounded reachability in PTAs. The complexity of model checking PTAs was studied in \cite{JurdzinskiLS07, LaroussinieS07}. The notion of bisimulation that we use was introduced in \cite{ChenHK08} and used for verification techniques in \cite{Sproston11}.  The computation and analysis of counterexamples in (non-probabilistic) timed automata was studied in~\cite{KolblLW2019,DierksKL2007}. Certification of unreachability was recently examined for timed automata \cite{WimmerM20}.
DBMs are a widely used data structure for timed systems (see~\cite{Tripakis98, KwiatkowskaNSW07}) that were first analyzed in \cite{Dill1990} and most notably used in the model checker UPPAAL\cite{Behrmann_etal06}.

\section{Preliminaries}\label{sec:prelims}

For any set $S$ we denote by $\Dist(S)$ the set of probability distributions on $S$ (seen as a discrete measurable space). Given $s\in S$, we let $\delta_s\in \Dist(S)$ denote the Dirac distribution on $s$, i.e. $\delta_s(t) = 0$ for all $t\neq s$ and $\delta_s(s) = 1$.

\subsubsection{Markov decision processes.}

A \emph{Markov decision process} (MDP) is a tuple $\M = (S, \Act, T, s_0)$, where $S$ is a set of \emph{states}, $\Act$ is a finite set of \emph{actions}, $T\colon S\to 2^{\Act\times \Dist(S)}$ is a \emph{transition function}, and $s_0\in S$ is the \emph{initial state}. We assume that $T(s)$ is non-empty and finite for all $s\in S$. A finite path is a sequence $\pi = s_0(\alpha_0, \mu_0)s_1(\alpha_1, \mu_1)... s_n$ such that for all $0\leq i\leq n-1$ we have $(\alpha_i, \mu_i)\in T(s_i)$ and $\mu_i(s_{i+1})>0$. A \emph{scheduler} $\S$ selects for each such finite path $\pi$ in $\M$ an element of $T(s_n)$. Infinite paths are defined accordingly. For $s\in S$ and $G\subseteq S$ the supremum $\prb_{\M, s}^{\max}(\lozenge G) := \sup_{\S} \; \Pr_{\M, s}^\S(\lozenge G)$ and infimum $\prb_{\M, s}^{\min}(\lozenge G):= \inf_{\S} \; \Pr_{\M, s}^\S(\lozenge G)$, ranging for all schedulers $\S$ over the probability of those $\S$-paths starting in $s$ and eventually reaching $G$, are attained (see, for example, \cite[Lemmata 10.102 and 10.113]{BaierK2008}).
We define $\prb^{*}_{\M}(\lozenge G) = \prb^{*}_{\M,s_0}(\lozenge G)$ for $* \in \{\min,\max\}$.
Let $\M = (S_\all, \Act, T, s_0)$ be an MDP with two distinguished absorbing states $\goal$ and $\fail$. A \emph{(weak) subsystem} $\mathcal{M}' $ of $\mathcal{M}$, denoted $\mathcal{M}'\subseteq\mathcal{M}$, is an MDP $\M' = (S'_\all,\Act,T', s_0)$ with $\goal, \fail \in S'_\all\subseteq S_\all$, and for each $(\alpha,\mu')\in T'(s)$ there exists $(\alpha, \mu) \in T(s)$ such that for $v\neq\fail$ we have $\mu'(v) \in \{0, \mu(v) \}$. Intuitively, in a subsystem some states and actions of $\M$ are deleted and some edges are redirected to $\fail$. A subsystem is \emph{strong} if, vice versa, for each $(\alpha, \mu) \in T(s)$ there exists $(\alpha,\mu')\in T'(s)$ with $\mu'(v) \in \{0, \mu(v)\}$. \footnote{This is a slight deviation from \cite{FunkeJB20}, where only strong subsystems were considered. Here we distinguish between weak and strong subsystems since it will reflect the corresponding notions for PTAs established in \Cref{sec:subsystem}.}

\subsubsection{Farkas certificates.} Let us assume that for all $s\in S := S_\all\setminus\{\goal, \fail\}$ we have $\prb^{\min}_s(\lozenge (\goal\lor\fail)) > 0$. In the following we write $\R^{\M}$ for the real vector space indexed by $ \bigcup_{s\in S} \{ s\}\times T(s)$. To each of the threshold properties $\prb_{s_0}^*(\lozenge\goal) \sim \lambda$ for $*\in\{\min, \max\}$ and $\sim \; \in \{ \leq, <,\geq, >\}$, one can associate a polytope (possibly with non-closed faces) sitting either in $\R^S$ or $\R^{\M}$ that is non-empty if and only if the threshold is satisfied. Elements in this polytope are called \emph{Farkas certificates} for the respective threshold property. The polytope of Farkas certificates for lower-bound thresholds $\prb_{s_0}^*(\lozenge\goal) \geq \lambda$ are of the form
\begin{gather*}
	\P^{\min}_{\M}(\lambda) = \{ \zb\in \R^{S}_{\geq 0} \mid \Ab\zb\leq \bb \land \zb(s_0) \geq \lambda \}, \; \text{ for } * = \min\\
	\P^{\max}_{\M}(\lambda) = \{ \yb\in \R^{\M}_{\geq 0}\mid \yb\Ab\leq \delta_{s_0} \land \yb\bb \geq\lambda \}, \; \text{ for } * = \max,
\end{gather*}
where $\Ab\in \R^{\M\times S}$ and $\bb\in\R^S$ can be taken as a black box in this paper. The main result of \cite{FunkeJB20} states that to any Farkas certificate $\zb\in\P_\M^{\min}(\lambda)$ one can associate a strong subsystem $\M'\subseteq\M$ whose states are contained in $\supp(\zb) = \{ s\in S\mid \zb(s) > 0\}$ and which satisfies $\prb_{\M', s_0}^{\min}(\lozenge\goal) \geq \lambda$. The corresponding statement holds for $\yb\in\P_\M^{\max}(\lambda)$ and subsystems with states contained in $\supp_S(\yb) = \{ s\in S\mid \exists \alpha\in T(s).\; \yb(s, \alpha) >0 \}$. 

\subsubsection{Clock constraints and difference bounds matrices.}

We fix a finite number of \emph{clocks} $\C = \{ c_0, c_1, ..., c_n\}$, where by convention $c_0$ is a designated clock always representing $0$ so that absolute and relative time bounds can be written in a uniform manner. A \emph{valuation} on $\C$ is a map $v\colon \C\to\R_{\geq 0}$ such that $v(c_0) = 0$. The set of all valuations on $\C$ is denoted by $\Val(\C)$. For a valuation $v$ and $t\in\R_{\geq 0}$ we denote by $v+t$ the valuation with $(v+t)(c) = v(c) + t$ for all $c\in \C\setminus\{c_0\}$. Given $C\subseteq\C$ we let $v[C:=0]$ be the \emph{reset} valuation with $v[C:=0](c) = 0$ for $c\in C$ and $v[C:=0](c) = v(c)$ for $c\notin C$. 
The set of \emph{clock constraints} $\CC(\C)$ is formed according to the following grammar: $g ::= \true \;|\; \false \;|\; c-c' \sim x \;|\; g\land g$,
where $c,c'\in \C$, $x\in \mathbb{Z} \cup \{\infty, -\infty\}$, and $\sim \; \in \{ \leq, <,\geq, >\}$. A valuation $v$ satisfies a clock constraint $g$, written as $v\models g$, if replacing every clock variable $c$ in $g$ with the value $v(c)$ leads to a true formula. We set $\Val(g) = \{ v \in \Val(\C)\mid v\models g\}$ and define $g_1 \Implies g_2$ if $\Val(g_1) \subseteq \Val(g_2)$.
A subset $Z\subseteq\Val(\C)$ is a \emph{zone} if $Z = \Val(g)$ for some clock constraint $g$. We commonly represent a clock constraint by a \emph{difference bounds matrix} (DBM), which is a $\C\times\C$-matrix $M$ over $(\mathbb{Z} \cup \{\infty,- \infty\}) \times \{<,\leq\}$. The intended meaning of an entry $M_{ij} = (a,\triangleleft)$ is the constraint $c_i - c_j \triangleleft a$. To each DBM $M$ one can associate a DBM $M^*$ containing constraints that are as tight as possible while still satisfying $\Val(M^*)= \Val(M)$ (see \cite[Theorem 2]{Dill1990}). We make use of the operations $ \sqcap $ from \cite{Dill1990} (corresponding to logical conjunction of the associated clock constraints) and the \emph{time closure} operation $\timeup$ of \cite{BengtssonW04} (there called \texttt{up}), which removes all absolute time bounds from the DBM, see also \citeapp{lem:timeup}.

\subsubsection{Probabilistic timed automata.}

A \emph{probabilistic timed automaton} (PTA) is a tuple $\T = (\Loc, \C, \Act, \inv, T, l_0)$, where $\Loc$ is a finite set of \emph{locations}, $\C$ is a finite set of \emph{clocks}, $\Act$ is a finite set of \emph{actions}, $\inv\colon\Loc\to \CC(\C)$ is the \emph{invariance condition}, $T\colon \Loc \to 2^{\CC(\C)\times\Act\times\Dist(2^\C\times\Loc)}$ is the transition function with $T(l)$ non-empty and finite for every $l\in\Loc$, and $l_0\in \Loc$ is the \emph{initial location}, for which we assume that $0\models\inv(l_0)$. A transition $(g, \alpha, \mu)\in T(l)$ is written as $l\overset{g:\alpha}{\longrightarrow}\mu$ and the element $g$ is called the \emph{guard}.
The intended meaning of $T(l)$ is that from location $l$ one first chooses non-deterministically a transition $l\overset{g:\alpha}{\longrightarrow}\mu$, provided that the guard $g$ is satisfied by the current clock valuation. Then an element $(C, l') \in 2^\C\times\Loc$ is picked according to the distribution $\mu$, the clocks in $C$ are reset and the next location is set to $l'$.

A \emph{timed probabilistic system} (TPS) is a tuple $\Scal = (S, \Act', T, s_0)$, where $S$ is a set of states, $\Act'=\Act\uplus\: \R_+$ is a set of actions ($\Act$ is assumed to be finite), $T: S\to 2^{\Act'\times\Dist(S)}$ is the transition function, and $s_0$ the \emph{initial state}. For a pair $(\alpha, \mu)\in T(s)$ (or $s\overset{\alpha}{\longrightarrow}\mu$) we assume that $\mu$ has finite support. Transitions indexed by $\R_+$ are called \emph{time delays} and transitions indexed by $\Act$ are \emph{discrete actions}. Schedulers are defined as for MDPs, and a scheduler $\S$ is \emph{time-divergent} if for almost every path compatible with $\S$ the series of time delays is divergent. Reachability probabilities $\prb_{\Scal, s}^*(\lozenge T)$ for $*\in\{\min, \max\}$ are defined as for MDPs, but only taking time-divergent schedulers into account.

A \emph{pointed} PTA $(\T, \goal, \fail)$ consists of a PTA $\T= (\Loc, \C, \Act, \inv, T, l_0)$ and two distinguished absorbing locations $\goal, \fail\in\Loc$. The \emph{semantics} of a pointed PTA is the TPS $\Scal(\T) = (S, \Act', T_{\sem}, s_0)$ with $S = \{ (l, v) \in \Loc\times\Val(\C)\mid v\models \inv(l)\}$, $\Act' = \Act\uplus\:\R_+$, $s_0 = (l_0, 0)$, and $T_{\sem}$ is the smallest function satisfying the inference rules
\[ 	\infer[\quad\text{and}]
{(l,v)\overset{t}{\longrightarrow}{\delta_{(l, v+t)}} \in T_{\sem}}
{t\in \R_+, \forall t'\leq t.\;\; v+t'\models\inv(l)}\quad\quad\;
	\infer[\text{, where}]
{(l,v)\overset{\alpha}{\longrightarrow}{\mu_\sem} \in T_{\sem}}
{l\overset{g:\alpha}{\longrightarrow}\mu \in T,\; v\models g}
\]
\begin{align}
	\mu_\sem(l',v') &= \sum_{\substack{(C,l') \\ v' = v[C:=0]}}  \;\mu(C, l') \quad \text{ for } l'\neq\fail\text{ and } v'  \models \inv(l')\\
	\mu_\sem(\fail, v') &= \sum_{\substack{(C, \fail)\\ v' = v[C:=0]}} \;\mu(C, \fail) +  \sum_{\substack{(C,l'), \;l'\neq\fail \\v' = v[C:=0]\not\models\inv(l')}} \;\mu(C, l') \label{eq:fail weight}
\end{align}

We define the goal set of $\Scal(\T)$ to be $\goal_{\Scal(\T)} = \{ (l,v)\in S\mid l = \goal\}$. For $*\in\{\min,\max\}$ the probability to reach $\goal$ in $\T$ is defined as
\[ \prb^*_{\T, l_0}(\lozenge\goal) := \prb^*_{\Scal(\T), s_0}(\lozenge\goal_{\Scal(\T)}) \]

\begin{remark}
	Typically, the semantics is only defined if the PTA is \emph{well-formed}. This means that no transition leads to a violation of the invariance condition of the target. We relax this condition and, in the case that $v' = v[C:=0]\not\models \inv(l')$, add the probability of $(C, l')$ to the edge  $(l,v)\overset{\alpha}{\longrightarrow} (\fail, v')$ (this is the second sum in \Cref{eq:fail weight}). This generalization will facilitate our translation from Farkas certificates of quotients of $\Scal(\T)$ to PTA subsystems.
\end{remark}

\subsubsection{Probabilistic time-abstracting bisimulation.} As in~\cite{ChenHK08}, we define
 a \emph{probabilistic time-abstracting bisimulation} (PTAB) on a TPS $\Scal = (S, \Act\uplus\:\R_{+}, T, s_0)$ to be an equivalence relation $\sim$ on $S$ such that if $s\sim s'$ we have:
\begin{enumerate}
	\item for any time delay $s\overset{t}{\to} u$ there exists a time delay $s'\overset{t'}{\to} u'$ such that $u\sim u'$;
	\item for any discrete action $s\overset{\alpha}{\to}\mu$, there exists a discrete action $s'\overset{\alpha}{\to} \mu'$ such that for all $E \in S/_{\sim}$ we have $\sum_{s\in E} \mu(s) = \sum_{s\in E} \mu'(s) $.
\end{enumerate}
If $\Scal$ has distinguished sets $\goal, \fail\subseteq S$, we say that a PTAB $\sim$ \emph{respects} $\goal$ and $\fail$ if whenever $(l,v) \sim (\goal,v')$, then $l = \goal$, and likewise for $\fail$. The \emph{quotient} of $\Scal$ by $\sim$ is the MDP $\M(\Scal/_{\sim}) = (S/_{\sim},\Act \cup \{\tau\},T', [s_0])$ with
  \begin{equation*}
    T'([s]) = \{(\tau,\delta_{[s']}) \mid \exists (t, \delta_{s'}) \in T(s)\} \;\cup\; \{(\alpha,\mu/_{\sim}) \mid \exists (\alpha,\mu) \in T(s)\}
  \end{equation*}
  with $\mu/_{\sim}(E') = \sum_{s' \in E'} \mu(s')$. As we could not find a formal proof for the following lemma in the literature, we included one in \referapp.
\begin{restatable}{lemma}{bisimreach}
  \label{lem:bisimreach}
  Let $\Scal$ be a TPS and $\sim$ a PTAB on $\Scal$ that respects $\goal$ and $\fail$. Then for all $s\in S$ and $*\in\{\min, \max\}$ we have
  \[\prb^{*}_{\Scal,s}(\lozenge \goal) = \prb^{*}_{\M(\Scal/_{\sim}),[s]}(\lozenge \goal).\]
\end{restatable}

\section{Witnessing subsystems for reachability in PTAs}\label{sec:subsystem}

In this chapter we generalize the notion of subsystems formalized first for Markov chains in \cite{JansenAKWKB11} and MDPs in \cite{WimmerJAKB14} to PTAs. From now on we assume for all pointed PTAs $(\T, \goal, \fail)$ that the probability to eventually reach $\goal$ or $\fail$ is $1$ for each time-divergent scheduler over the semantics $\Scal(\T)$.
This is necessary to apply the results of~\cite{FunkeJB20}.
An important application that justifies this assumption is \emph{time-bounded} reachability, where $\goal$ needs to be reached before an absolute time-bound $K$.
This can be encoded in our setting by adding a clock $c^*$ that is never reset, and adding $c^*\leq K$ to the invariance of every location.

\subsection{Subsystems for PTAs}

\begin{definition}[Subsystem]\label{def:subsystem}
	Let $(\T, \goal, \fail)$ be a pointed PTA with $\T= (\Loc, \C, \Act, \inv, T, l_0)$.
  A PTA $\T'= (\Loc', \C, \Act, \inv', T', l_0)$ is a \emph{(weak) subsystem} of $\T$ if the following three conditions hold:
  \begin{enumerate}[align=left, leftmargin=1.2cm, labelwidth = 3.5ex]
  \item\label{item:loccontainment} $\goal, \fail\in \Loc'\subseteq\Loc$;
  \item\label{item:invcontainment} for all locations $l \in \Loc'$ we have $\inv'(l) \Implies \inv(l)$;
  \item\label{item:weak transition} for all $l\in\Loc'$ there is an injective map $\Phi\colon T'(l) \to T(l)$ such that for $\Phi(l\overset{g':\alpha'}{\longrightarrow}\mu') = l\overset{g:\alpha}{\longrightarrow}\mu$ we have (3a) $g' \Implies g$, (3b) $\alpha'=\alpha$, and (3c) for all $(C, l')\in 2^{\C}\times \Loc'$ with $l' \neq \fail$ we have  $\mu'(C,l') \in \{0, \mu(C, l')\}$.    
  \end{enumerate}
 We call $\T'$ a \emph{strong subsystem} if, additionally, the following two conditions hold for all $l\in \Loc'$:
\begin{enumerate}[leftmargin=1.2cm, labelwidth = 3.5ex, topsep = 2mm]
  {\setlength\itemindent{4.5pt}\item[(3$\,^*\!$)] there is a left-inverse $\Psi\colon T(l) \to T'(l)$ of $\Phi$ such that for $\Psi(l\overset{g:\alpha}{\longrightarrow}\mu) = l\overset{g':\alpha'}{\longrightarrow}\mu'$ we have (3a$\,^*\!$) $g' \equiv g \land \inv'(l)$, and (3b) and (3c) as above}\label{item:strong transition};
   \item[(4)] if $v \in \Val(\C)$ and $t \in \R_+$ satisfy $v\models \inv'(l)$ and $v+t\models\inv(l)$, then also $v+t\models\inv'(l)$.
  \end{enumerate}
\end{definition}

In other words, in the passage from $\T$ to a subsystem, it is allowed to discard locations and elements in $T(l)$, redirect individual transitions to $\fail$, and shrink invariants and guards. This will be sufficient for witnessing lower bounds on $\prb^{\max}$ (see \Cref{prop:monotonicity} below). For witnessing lower bounds on $\prb^{\min}$ we need the extra assumptions that elements in $T(l)$ must not be deleted, guards can only shrink as much as the invariance and that $\inv'(l)$ is closed under time successors. On the level of quotients of the semantics of $\T$, this reflects the difference between weak and strong subsystems for MDPs (see \Cref{sec:prelims}). We demand $\Psi$ to be a left-inverse of $\Phi$ instead of requiring that both are bijections since two different elements of $T(l)$ might coincidentally be shrunk to the same element of $T'(l)$.

\begin{example}
  \label{ex:subsystem}
	Consider the PTA $\T$ displayed in \Cref{fig:example_pta}. A scheduler $\S$ in $\T$ principally has to choose between $\alpha$ and $\beta$ whenever in $l_1$ (and letting time pass accordingly). Action $\alpha$ in state $(l_1, (x,y))\in \Scal(\T)$ leads to a higher probability to reach $\goal$ exactly when $y\leq 2$, the reason being that then the right-hand branch of $\T$ contributes towards $\Pr^\S(\lozenge \goal)$ upon leaving $l_0$ the next time. Thus choosing $\beta$ upon leaving $l_1$ for the first time leads to a scheduler attaining $\prb^{\min}_{\T}(\lozenge\goal)$ (cf. \Cref{ex:subsystem extended} in \referapp). An example of a weak subsystem $\T'\subseteq\T$ is portrayed in \Cref{fig:example_subsystem}, with differences to $\T$ indicated in red. Even though $\T'$ fails to be a strong subsystem (e.g. the guard of $\alpha$ is shrunk more than allowed), we have $\prb^{\min}_\T(\lozenge\goal) \geq \prb^{\min}_{\T'}(\lozenge\goal)$. However, this is not true for all weak subsystems: Take $\T''$ obtained from $\T$ by changing only the guard of the action $\beta$ at $l_1$ from $x\leq 1$ to $x\leq 1 \land y\geq 2$. Then any scheduler is forced to take $\alpha$ at least once, resulting in $\prb^{\min}_\T(\lozenge\goal) < \prb^{\min}_{\T''}(\lozenge\goal)$. Removing action $\beta$ and location $l_3$ altogether has the same effect. This example illustrates that strong subsystems are indeed needed in order to deal with $\prb^{\min}$ (cf. \Cref{prop:monotonicity}).
  More details can be found in \referapp.
		\input{pta_example.tex}
\end{example}

We show that subsystems of a PTA $\T$ induce Farkas certificates in finite-state quotients of $\Scal(\T)$, which are supported on the states induced by the subsystem. In other words, subsystems are reflected \emph{purely algebraically} on the level of Farkas certificates. This is a generalization of the forward direction of \cite[Theorem 5.4]{FunkeJB20}.

\begin{restatable}[PTA subsystems induce Farkas certificates]{theorem}{subsysandfarkascert}
  \label{prop:subsysandfarkascert}
  Let\linebreak $(\T, \goal, \fail)$ be a pointed PTA, and let $\sim$ be a PTAB on $\Scal(\T)$ that respects $\goal$ and $\fail$ and has finite index.
  Let $\M = \M(\Scal(\T)/_\sim)$ be the associated quotient MDP with states $S\cup\{\goal, \fail\}$.
  Given a subsystem $\T'\subseteq\T$, let $S' = \{[s] \in S \mid  s \text{ is a state of } \Scal(\T') \}$.

  Then there is a Farkas certificate $\yb\in\mathbb{R}^{\M}$ for $\prb^{\max}_{\M}(\lozenge \goal) \geq \prb^{\max}_{\T'}(\lozenge \goal)$ with $\supp_S(\yb)\subseteq S'$.
  If $\T'$ is a strong subsystem, then there also exists a Farkas certificate $\zb\in\mathbb{R}^S$ for $\prb^{\min}_{\M}(\lozenge \goal) \geq \prb^{\min}_{\T'}(\lozenge \goal)$ such that $\supp(\zb)\subseteq S'$.
\end{restatable}

\begin{corollary}\label{prop:monotonicity}
  Let $(\T, \goal, \fail)$ be a pointed PTA.
  \begin{enumerate}
  	\item If $\T' \subseteq \T$ is a subsystem, then $\prb^{\max}_{\T}(\lozenge \goal) \geq \prb^{\max}_{\T'}(\lozenge \goal)$.
  	\item If $\T'\subseteq\T$ is a strong subsystem, then $\prb^{\min}_{\T}(\lozenge \goal) \geq \prb^{\min}_{\T'}(\lozenge \goal)$.
  \end{enumerate} 
\end{corollary}

\subsection{Zone closure for DBMs}\label{sub:zone closure}

Our next aim is to show how Farkas certificates of the quotient $\M(\mathcal{S}/_{\sim})$ can be translated back into PTA subsystems.
As location invariants are described by zones, this requires to pass from states of the quotient (which represent equivalence classes of clock valuations) to zones that include these valuations and are as small as possible.
We do this using the following operation, which relies on the lexicographic order on DBMs (see also \citeapp{sub:dbms}).

\begin{definition}[Zone closure]\label{def:zone closure}
	Let $M$ and $N$ be DBMs over $\C$.
The \emph{zone closure} $M \sqcup N$ is the DBM defined by
		\[ (M \sqcup N)_{ij} = \max \{M_{ij},N_{ij}\} \;\text{ for all } i, j\in \C.\]
\end{definition}
The zone closure satisfies the following properties:
\begin{restatable}{lemma}{zoneclosure}
	\label{lem:maxunion}
	\label{lem:minpreserv}
	Let $M, N$ be DBMs such that $M = M^*$ and $N = N^*$.
	Then
  \begin{enumerate}
  \item $\Val(M \sqcup N)$ is the smallest zone in $\Val(\C)$ containing $\Val(M) \cup \Val(N)$.
  \item We have $(M \sqcup N)^* = (M \sqcup N)$.
  \end{enumerate}
 \end{restatable}

Given an arbitrary subset $R\subseteq \Val(\C)$ the \emph{canonical DBM} $M_R$ associated to $R$ is defined as $(M_{R})_{ij} = (\sup \{v(i) - v(j) \mid v \in R\},\triangleleft)$ for $i, j\in\C$, where $\triangleleft \; = \; \leq$ exactly if the supremum is attained, and otherwise $<$. Then $M_R = M_R^*$ and $\Val(M_R)$ is the smallest zone of $\Val(\C)$ that contains $R$ (see \citeapp{lem:canonicalDBM}). Applying~\Cref{lem:maxunion} to the canonical DBM associated to sets of clock valuations gives:
\begin{restatable}{proposition}{regioncup}
  \label{lem:regioncup}
  Let $R_1, ..., R_n\subseteq \Val(\C)$ be sets of clock valuations.
  For every $i$ let $M_{R_i}$ be the canonical DBM of $R_i$ and set $M = \bigsqcup_{i=1}^n \;M_{R_i}$.
  Then, $\Val(M)$ is the smallest zone in $\Val(\C)$ that contains all sets $R_i$.
\end{restatable}

\subsection{From Farkas certificates to witnessing subsystems}

We are now in a position to outline a construction which reverses \Cref{prop:subsysandfarkascert}, i.e., which passes from Farkas certificates for threshold properties in finite-state quotients of the PTA semantics to PTA subsystems. Of course, the constructed subsystems should \emph{witness} the same threshold on the level of the PTA, as follows:

\begin{definition}[Witness]\label{def:witness}
	Let $(\T, \goal, \fail)$ be a pointed PTA and let $\lambda\in[0,1]$.
	A \emph{witnessing subsystem} or simply a \emph{witness} for $\prb^{\max}_{\T}(\lozenge \goal) \geq\lambda$ is a subsystem $\T'\subseteq \T$ such that $\prb^{\max}_{\T'}(\lozenge \goal) \geq \lambda$.
	A \emph{witnessing subsystem} or \emph{witness} for $\prb^{\min}_{\T}(\lozenge \goal) \geq\lambda$ is a \emph{strong} subsystem $\T'\subseteq \T$ such that $\prb^{\min}_{\T'}(\lozenge \goal) \geq\lambda$.
\end{definition}

By \Cref{prop:monotonicity} a witnessing subsystem is indeed a witness for the given threshold property.
The next definition shows how to construct a witness from Farkas certificates of finite-state quotients of the PTA semantics. Here and for the rest of this section we use the notation $S = S_{\all} \setminus \{\goal,\fail\}$, where $S_\all$ are the states of a PTAB quotient of $\Scal(\T)$.

\begin{definition}[Induced subsystems]
  \label{def:indsubsys}
	Let $(\T, \goal, \fail)$ be a pointed PTA, and let $\M = (S_\all, \Act, T, s_0)$ the quotient of $\Scal(\T)$ by a PTAB $\sim$ that respects $\goal$ and $\fail$ and has finite index. Given $s\in S$ and $l\in\Loc$ we put
	\[ s_{|l} = \{v\in\Val(\C)\mid (l,v) \in s\}.\]
	For a fixed $R\subseteq S$ we define subsystems $\T_R^w= (\Loc', \C, \Act, \inv^w, T^w, l_0)$ and $\T_R^s = (\Loc', \C, \Act, \inv^s, T^s, l_0)$ induced by $R$ as follows:
\begin{itemize}
\item Both have locations: $\Loc' = \{l \in \Loc \mid \exists s\in R. \; s_{|l}\neq\emptyset\}\cup \{\goal, \fail\}$
\item For each location $l \in \Loc'$ we consider the DBMs
  \[M_l^w = \bigsqcup_{s \in R} M_{s_{|l}} \;\;\;\;\text{ and }\;\;\;\; M_l^s = (\timeup M_l^w) \sqcap M_{\inv(l)}\]
  and let $\inv^w(l) = M_l^w$ and $\inv^s(l) = M_l^s$.
\item For every $l\overset{g : \alpha}{\longrightarrow}\mu$ in $T(l)$ with $l\in \Loc'$ let
\[g^w = g\sqcap\bigsqcup_{\substack{s \in R\\ \exists (l,v) \in s.\: v\models g}} M_{s_{|l}} \;\;\;\;\text{ and }\;\;\;\; g^s = g \sqcap \inv^s(l)\]

For $C\subseteq \C$ and $l'\in \Loc'\setminus\{\fail\}$ let
\[
\mu'(C, l') = \begin{cases}
\mu(C, l') & \text{ if }\:\exists s, s'\in R, (l, v)\in s.\; (l',v[C := 0]) \in s' \\
0 & \text{ otherwise }
\end{cases}
\]
and assign the remaining probability to $\mu'(\fail, \emptyset)$. Now add a transition $l\overset{g^w : \alpha}{\longrightarrow}\mu'$ to $T^w(l)$ and $l\overset{g^s : \alpha}{\longrightarrow}\mu'$ to $T^s(l)$.
  \end{itemize}
\end{definition}

The intuition behind this construction is that one completes all states of $\T$ whose equivalence class is in $R$ to a smallest (weak or strong) subsystem of $\T$ whose state space contains this set. In each location, the set of clock valuations which induce states in $R$ is turned into a viable invariance condition using the operation $\sqcup$. Guards of transitions in $\T$ are shrunk accordingly, and their support is restricted to those pairs $(C, l')$ which -- on the level of the quotient $\M$ -- induce at least one transition between two elements of $R$.

\begin{restatable}{lemma}{rtoptasubsset}
  \label{lem:rtoptasubsset}

  Let $(\T, \goal, \fail)$ be a pointed PTA and $\M = (S_\all, \Act, T, s_0)$ the quotient of $\Scal(\T)$ by a PTAB that respects $\goal$ and $\fail$. Then for any $R\subseteq S$, $\T_R^w$ is a subsystem and $\T_R^s$ is a strong subsystem of $\T$.
\end{restatable}

The following proposition states that Farkas certificates for any PTAB quotient of the PTA can be used to find witnesses for probabilistic reachability constraints. 
It is a generalization of the backward direction of \cite[Theorem 5.4]{FunkeJB20} and provides a converse of \Cref{prop:subsysandfarkascert}.

\begin{restatable}[Farkas certificates to witnesses]{proposition}{witsubsysandfarkcert}
  \label{prop:witsubsysandfarkcert}
	Let $(\T, \goal, \fail)$ be a pointed PTA and $\M = (S_\all, \Act, T, s_0)$ the quotient of $\Scal(\T)$ by a PTAB $\sim$ that respects $\goal$ and $\fail$. Fix $\lambda\in[0,1]$ and $R\subseteq S$.
	
	If there exists a Farkas certificate $\zb\in\P^{\min}_{\M}(\lambda)$ with $\supp(\zb)\subseteq R$, then $\T_R^s$ is a witness for $\prb_{\T}^{\min}(\lozenge \goal) \geq \lambda$. Likewise, if there exists a Farkas certificate $\yb\in\P^{\max}_{\M}(\lambda)$ with $\supp_S(\yb)\subseteq R$, then $\T_R^w$ is a witness for $\prb_{\T}^{\max}(\lozenge \goal) \geq \lambda$.
\end{restatable}

\section{Computing minimal witnessing subsystems}\label{sec:minimal}

We now introduce three notions of minimality for subsystems of PTAs and show how minimal (or small) subsystems can be computed.
Henceforth let $\M$ be the quotient (with states $S_\all$) of the semantics of a pointed PTA ($\T, \goal, \fail)$ by a PTAB $\sim$ that has finite index and let $S = S_{\all} \setminus \{\goal,\fail\}$.

As the threshold problem for $\min$ and $\max$-reachability constraints of PTAs is directly reducible to the \emph{existence} of a witness for the same property, computing (minimal) witnessing subsystems is at least as hard as this problem.
Deciding $\prb^{\max}_{\T}(\lozenge \goal) \geq 1$ is \EXPTIME-hard~\cite[Theorem 3.1]{LaroussinieS07} for PTAs, which holds already for time-bounded reachability.
\PSPACE-hardness of $\prb^{\min}_{\T}(\lozenge \goal) \geq 1$ (which is equivalent to $\prb^{\max}_{\T}(\lozenge \goal) > 0$ in the time-bounded setting) follows from \PSPACE-hardness of non-probabilistic reachability~\cite[Theorem 4.17]{AlurD1994}.

\subsection{Notions of minimality for PTA subsystems}

For a set of valuations $R\subseteq \Val(\C)$ we denote by $\vol(R)$ the Lebesgue volume of $R$ considered as a subset of $\mathbb{R}^{\C\setminus\{c_0\}}$. The \emph{volume} of a PTA $\T$ is defined as 
\[\vol(\T) = \sum_{l\in \Loc(\T)} \; \vol \big(\Val(\inv(l))\big) \in \R_{\geq 0} \cup \{\infty\}.\]
\vspace{-2mm}
\begin{definition}[Notions of minimality]
	We define three partial orders on subsystems $\T_1,\T_2$ of a PTA $\T$ as follows:
	\begin{enumerate}
		\item $\T_1\leq_{\loc}\T_2$ if $|\Loc(\T_1)| \leq |\Loc(\T_2)|$;
		\item $\T_1\leq_{\inv}\T_2$ if $\Loc(\T_1)\subseteq\Loc(\T_2)$ and for all $l \in \Loc(\T_1): \; \inv_{\T_1}(\l) \Implies \inv_{\T_2}(l)$;
		\item $\T_1\leq_{\vol}\T_2$ if $\vol(\T_1) \leq \vol(\T_2)$.
	\end{enumerate}
	We say that a witness $\T'\subseteq\T$ for some threshold property as defined in \Cref{def:witness} is \emph{loc-minimal} (respectively, \emph{inv-minimal} or \emph{vol-minimal}) if $\T'$ is a $\leq_{\loc}$-minimal element (respectively, $\leq_{\inv}$-minimal or $\leq_{\vol}$-minimal element) among all witnesses of $\T$ for the same threshold property.
\end{definition}

When considering inv- and vol-minimality, we will assume that $\Val(\inv(l))$ is bounded for every location $l\in \Loc$, or, equivalently, that a finite upper bound on all clocks exists.
This will guarantee that the set of witnesses that we have to consider is finite, and, for vol-minimality, that their volume is finite. 

The rationale for considering vol-minimal witnesses is that they have -- in a precise measure-theoretic sense -- a minimal number of states. Note that in contrast to $\leq_{\loc}$ and $\leq_{\vol}$, the partial order $\leq_{\inv}$ is not a total order and thus results in general in many incomparable inv-minimal witnesses.

\renewcommand{\arraystretch}{1.2}
\renewcommand{\tabcolsep}{2mm}
\setlength{\textfloatsep}{0.3cm}

\begin{example}
  \label{ex:minimal}
	Consider the PTA of \Cref{ex:subsystem}. \Cref{tab:overview} lists minimal witnesses for $\lambda = 6/25$ for all three notions of minimality.
    The inv-minimal witnesses for $\prb^{\max}$ also encode corresponding schedulers with probability of at least $6/25$ (e.g. the first one encodes waiting in $l_1$ for one time unit, choosing $\alpha$, and on the branch going through $l_0$ repeating this once more).
    For $\prb^{\min}$, the inv-minimal witnesses ensure that whatever choice the scheduler makes the induced probability will be at least $6/25$. See \Cref{ex:minimal extended} in \referapp{} for more details.
	\begin{table}[h]
		\centering
		\captionsetup{labelfont={up}}
		\caption{Every indent describes a minimal witness for the PTA $\T$ in \Cref{fig:example_pta}. For inv-minimal ones, invariants are highlighted in blue after colons of the corresponding location, where the clock $x$ is drawn on the horizontal axis, $y$ on the vertical axis, and gridlines have unit $1$.}		
		\begin{tabular}{|p{5mm}|p{5.0cm}|p{5.2cm}|}
			\hline
			& \multicolumn{1}{c|}{$\prb^{\max}_\T(\lozenge\goal)\geq 6/25$}
			& \multicolumn{1}{c|}{$\prb^{\min}_\T(\lozenge\goal)\geq 6/25$}
			\\ \hline
			\vspace{1mm}  loc
			&
			\begin{itemize}[topsep=-5ex, itemsep=0.5ex]
				\item keeping $l_0$ and $l_1$;
				\item keeping $l_0$ and $l_2$;\vspace*{-0.6\baselineskip}
			\end{itemize} 
			& 
			\begin{itemize}[topsep=-15ex]
				\item keeping $l_0$ and $l_2$;
			\end{itemize} 
			\\\hline
			\vspace{15mm}\centering inv 
			& \begin{itemize}[topsep=-5ex, itemsep=1ex,  leftmargin = 2.5ex]
					\item $l_0$: \input{line_vertical.tex}, $l_1$: \input{parallelogram.tex}
					\item $l_0$: \input{point_00.tex}, $l_2$: \input{line.tex}
					\item $l_0$: \input{point_00.tex}, $l_1$: \input{point_00.tex}, $l_3$: \input{line.tex}
					\item $l_0$: \input{point_00.tex}, $l_1$: \input{line.tex}, $l_3$: \input{point_11.tex}\vspace*{-0.7\baselineskip}
				\end{itemize} 
			&  
				\begin{itemize}[topsep=-5ex, itemsep=1.2ex,   leftmargin = 2.5ex]
					\item $l_0$: \input{point_00.tex}, $l_2$: \input{long_line.tex}
					\item $l_0$: \input{long_line_vertical.tex}, $l_1$: \input{long_parallelogram.tex}, $l_3$: \input{vertical_parallelogram.tex}\vspace*{-0.6\baselineskip}
				\end{itemize}
			\\\hline
			\vspace{1mm} vol
			& 	
				\begin{itemize}[topsep=-5ex, leftmargin = 2.5ex]
					\item the bottom three inv-minimal witnesses from above ($\vol = 0$)\vspace*{-0.55\baselineskip}
				\end{itemize}
			& 
				\begin{itemize}[topsep=-5ex,  leftmargin = 2.5ex]
					\item the top inv-minimal witness from above ($\vol=0$)\vspace*{-0.55\baselineskip}
				\end{itemize}
			\\ \hline
		\end{tabular}
	\label{tab:overview}
	\end{table}
\end{example}

\begin{restatable}{lemma}{partialorders}\label{lem:orders}
  We have $\leq_{\inv} \; \subseteq \; \leq_{\loc} \cap \leq_{\vol}$. Moreover, $\leq_{\vol}$ and $\leq_{\loc}$ are incomparable in general.
\end{restatable}

Note that \Cref{lem:orders} does not imply that inv-minimal witness are loc-minimal or vol-minimal. This is because an inv-minimal witness might be $\leq_{\inv}$-incomparable to witnesses with smaller volume (see also \Cref{ex:minimal}).

\subsection{Computing loc-minimal witnesses}

In this section we will assume that whenever $(l_1,v_1) \sim (l_2,v_2)$, then $l_1 = l_2$.
To compute a loc-minimal strong subsystem of $\T$ we use a mixed integer linear program (MILP) over the inequalities defining $\P^{\min}_{\M}(\lambda)$ (see~\Cref{sec:prelims}). We first define the linear inequalities:
\begin{equation}\label{eq:loc constr}
  \tag{\footnotesize \texttt{LOC-CONSTR}}
  \zb \in \P_{\M}^{\min}(\lambda)\; \text{ and }\; \zb_{[(l,v)]} \leq \zeta_l \text{ for all } [(l,v)] \in S
\end{equation}
This adds exactly $|S|$ inequalities to the ones defining $\P^{\min}_{\M}(\lambda)$. The idea is that as the variable $\zb_{[(l,v)]}$ measures whether $[(l,v)]$ should be contained in the MDP subsystem associated with a Farkas certificate, the new variable $\zeta_l$ measures whether location $l$ is needed \emph{at all} in the corresponding PTA subsystem.

\begin{restatable}{proposition}{locminimal}
  \label{prop:locminimalws}
	There exists a witnessing subsystem for $\prb^{\min}_{\T}(\lozenge \goal) \geq\lambda$ with at most $k$ locations (excluding $\goal$ and $\fail$) if and only if there exists a pair $(\zb, \zeta)$ that satisfies \emph{(\ref{eq:loc constr})}, where $\zeta$ has at most $k$ non-zero entries.
\end{restatable}
Restricting $\zeta_l$ to the domain $\{0,1\}$ leads to the following MILP:
\begin{equation}\label{eq:loc milp}
\tag{\footnotesize\texttt{LOC-MILP}}
\min \sum_{l \in \Loc} \zeta_l \;\; \text{ s.t. } \;\; (\zb,\zeta) \text{ \footnotesize satisfies (\ref{eq:loc constr})}
\end{equation}
By~\Cref{prop:locminimalws}, solutions of (\ref{eq:loc milp}) correspond to loc-minimal witnesses for $\prb^{\min}_{\T}(\lozenge \goal) \geq\lambda$. Although the size of (\ref{eq:loc milp}) is exponential in the size of $\T$, it has only $|\Loc|$ many binary variables.
Hence, if the size of $\M$ is single-exponential (as is already the case for the \emph{region graph}, see~\cite{AlurCD93,KwiatkowskaNSS02}), a loc-minimal witness can be computed in single-exponential time:

\begin{restatable}{proposition}{locminimalmilp}
  \label{prop:locminimalmilp}
  A loc-minimal witness for $\prb^{\min}_{\T}(\lozenge \goal) \geq\lambda$ can be computed in time $\mathcal{O}(2^{|\Loc|} \cdot \poly(|\M|))$, if one exists.
\end{restatable}
One can deal with $\prb^{\max}_{\T}(\lozenge \goal) \geq\lambda$ similarly.
In~\cite{FunkeJB20} the \emph{quotient sum heuristic} was introduced as an approach for finding vectors with many zeros in a given polytope by iteratively solving LPs whose objective function is the inverse of the last optimal solution. This approach can be adapted to maximize zeros in only part of the dimensions by assigning the objective value $0$ to the rest. In the case of loc-minimal witnesses one discards all variables $\zb_{[(l,v)]}$ and optimizes only over the new variables $\zeta_l$ (which are non-binary in the LP-based QS heuristic).

\subsection{Computing inv-minimal witnesses}
We now assume that $\Val(\inv(l))$ is bounded in every location $l$, and take $K$ to be an upper bound on all clocks that must then exist. While for loc-miminality we assumed that $\sim$ distinguishes locations, now we additionally assume that if $(l_1,v_1) \sim (l_2,v_2)$, then there is no clock constraint $\gamma$ such that $v_1 \models \gamma$ and $v_2 \not\models \gamma$. So, equivalent valuations must be indistinguishable by clock constraints.
The coarsest PTAB that achieves this is the region equivalence (see~\cite{AlurCD93,KwiatkowskaNSS02}).

To encode invariance strength, we will use $n = 4K{+}1$ binary variables $\xi_{ij}^l(k)$ with $k \in \{-2K,\ldots,2K\}$ for every location $l$ and ordered pair of clocks $c_i,c_j$.
The intended meaning of $\xi_{ij}^l(k) = 1$ is that $\ceil*{k/2}$ is an upper bound for $v(i) - v(j)$ for all $v \in \Val(\inv(l))$.
We have introduced the granularity $1/2$ in order to distinguish between strict and non-strict inequalities.
For even $k$, which will represent $\leq$, the upper bound will always be met. 
Formally, we consider the following constraints, ranging over $l\in\Loc$ and $c_i, c_j\in \C$ with $j\neq 0$:
\begin{align}\label{eq:inv constr}
\begin{split}
&\zb \in \P_{\M}^{\min}(\lambda) \\
&\zb_{[(l,v)]}  \;\,\leq \quad
\begin{cases}
	\xi_{ij}^l(2a{-}1) & \text{ if } (M_{[(l,v)]})_{ij} = (a, <) \\
	\xi_{ij}^l(2a) & \text{ if } (M_{[(l,v)]})_{ij} = (a, \leq)
\end{cases} \\
  &\xi_{ij}^l(k) \;\: \leq \;\:\xi_{ij}^l(k{-}1)\quad {\footnotesize \text{ for all } k \in \{-2K{+}1,\ldots,2K\} }  
  \end{split}
  \tag{{\footnotesize \texttt{INV-CONSTR}}}
\end{align}
In the above, $M_{[(l,v)]}$ is the canonical DBM for the set of valuations $\{v' \in \Val(\C)\mid (l,v') \in [(l,v)]\}$ as defined in~\Cref{sub:zone closure}. The reason for excluding the constraints where $c_j$ is the zero clock is that for strong subsystems a stronger invariant cannot be achieved by strengthening the upper bound of a clock, cf. \Cref{def:subsystem}, (4). On top of these constraints we now define the MILP:
\begin{equation}\label{eq:inv milp}
\tag{\footnotesize \texttt{INV-MILP}} 
\min \,\sum_{l,i,j,k} \xi_{ij}^l(k) \; \text{ s.t } \; (\zb,\xi) \text{ satisfies (\ref{eq:inv constr})}. 
\end{equation}
\begin{restatable}{proposition}{invminimal}
  If $(\zb,\xi)$ is a solution of \emph{(\ref{eq:inv milp})}, then $\T_{\supp(\zb)}^s$ is an inv-minimal witness for $\prb_{\T}^{\min}(\lozenge \goal) \geq \lambda$.
\end{restatable}

The number of binary variables in (\ref{eq:inv milp}) is $n \cdot |\Loc| \cdot (|\C|^2-|\C|)$.
However, due to the constraints $\xi_{ij}^l(k) \leq \xi_{ij}^l(k{-}1)$, there are only $n$ possible configurations of the binary variables $\xi_{ij}^l(k)$ for every location $l$ and pair of clocks $c_i,c_j$.
Hence, the number of satisfying configurations of $\xi$ is bounded by $n^{|\Loc| \cdot (|\C|^2-|\C|)}$.
In a similar way as for~\Cref{prop:locminimalmilp} we get:
\begin{proposition}
  An inv-minimal witness for $\prb^{\min}_{\T}(\lozenge \goal) \geq \lambda$ can be computed in time $\mathcal{O}(2^{\log(n) \cdot |\Loc| \cdot |\C|^2} \cdot \poly(|\M|))$, if one exists.
\end{proposition}
Again, $\prb^{\max}_{\T}$ can be treated similarly and the same idea of deriving heuristics that was outlined to loc-minimal witnesses can be used here.

\subsection{Computing vol-minimal witnesses}

As for inv-minimality, we will assume that $\sim$ distinguishes states that are distinguishable by clock constraints and that $K$ is an upper bound on all clocks.
To get a candidate set of possible vol-minimal witnesses, we use the following lemma:
\begin{restatable}{lemma}{invvolnotdisjoint}
	For $*\in \{\min, \max\}$, there is at least one witness for $\prb^{*}_{\T}(\lozenge \goal) \geq \lambda$ that is both inv- and vol-minimal.
\end{restatable}

Hence, to find a vol-minimal witness it suffices to compute (1) \emph{all} inv-minimal witnesses and (2) compare their volumes. Using the results of the previous section, for (1) it is enough to solve the \emph{multi-objective} mixed integer linear program
\begin{equation}
\label{eq:inv mo}
\tag{\footnotesize \texttt{INV-MO}} \text{ for all } \; {\substack{l \in \Loc \\ c_i,c_j \in \C \\ j \neq 0}}: \; \min \sum_k \xi_{ij}^l(k) \;\text { s.t. } \;(\zb,\xi) \text{ satisfies (\ref{eq:inv constr})} 
\end{equation}

A solution of this program is a vector that satisfies (\ref{eq:inv constr}) and such that all other vectors satisfying (\ref{eq:inv constr}) evaluate worse on at least one objective function. This implies that the set of solutions of (\ref{eq:inv mo}) encodes precisely the set of inv-minimal witnesses for $\prb^{\min}_\T(\lozenge\goal)\geq \lambda$. 
Techniques for solving such programs efficiently are presented in~\cite{PetterssonO2019,OzpeynirciK2010}. 

Let $\text{\textsc{vol}}(|\C|^2, \log(K))$ be the time complexity of computing the volume of a DBM over clocks $\C$ with entries bounded from above by $K$. This factor is exponential in general, but polynomial if the number of clocks is fixed \cite{GritzmannK1994}. Then we get the following time complexity for computing vol-minimal witnesses:

\begin{proposition}
	A vol-minimal witness for $\prb^{\min}_{\T}(\lozenge \goal) \geq \lambda$ can be computed in time $\mathcal{O}(2^{\log(n) \cdot |\Loc| \cdot |\C|^2} \cdot \text{\textsc{vol}}(|\C|^2,\log(K)) \cdot \poly(|\M|))$, if one exists.
\end{proposition}

\subsection{Hardness of deciding $\leq_{\vol}$}

Computing the volume of a polytope generally requires exponential time in the number of dimensions. However, as the invariants of PTA have a restricted form involving only linear inequalities with at most two clocks, one might hope that computing their volume is easier. We now show that this is not the case (under the standard complexity theoretic assumptions).

We recall that \hashP{} is the counting complexity class that includes the functions that can be expressed as the number of accepting runs of a polynomial time, non-deterministic Turing machine (NTM) for a given input.
Hardness for \hashP{} is typically defined using polynomial-time Turing reductions.
The analogous decision class is \PP{}, where $L \in $ \PP{} if there is a polynomial time NTM such that $x \in L$ if and only if the majority of runs of the NTM on $x$ is accepting (see~\cite[Chapter 9]{DBLP:books/daglib/0023084} for an introduction). Via a reduction from \hashP-hardness results on polytope volume computation, we obtain:

\begin{restatable}{proposition}{dbmhashp}
  \label{prop:dbmvolhashp}
  Computing $\vol(\Val(M))$ for a DBM M is \emph{\hashP}-hard.
\end{restatable}
Using this proposition we can show that deciding the $\leq_{\vol}$ relation for two PTA subsystems is substantially harder than for $\leq_{\loc}$ and $\leq_{\inv}$.
\begin{restatable}{theorem}{leqvolpp}\label{prop:PP}
   Given two subsystems $\T_1,\T_2$ in a PTA $\T$, deciding whether $\T_1 \leq_{\vol} \T_2$ holds is \emph{\PP}-hard under polynomial-time Turing reductions.
\end{restatable}

Hence, in particular, there is no polynomial time algorithm to decide $\T_1 \leq_{\vol} \T_2$, unless \texttt{P} $=$ \texttt{NP}.
This should be contrasted with the relations $\leq_{\loc}$ and $\leq_{\inv}$.
To decide $\T_1 \leq_{\loc} \T_2$ one just counts the locations, and for $\T_1 \leq_{\inv} \T_2$ one checks the inclusion of locations and inspects the entries of the canonical DBMs associated to the invariants.
In fact, these observations for $\leq_{\loc}$ and $\leq_{\inv}$ are the main ingredients for the MILP formulations (\ref{eq:loc milp}) and (\ref{eq:inv milp}).

\section{Conclusion} 

This paper introduces witnessing subsystems for PTAs.
These subsystems give insight into which (hopefully small) part of the system is sufficient for a certain property to hold.
We have studied three notions of minimality for witnessing subsystems: location number, invariant strength, and invariant volume. For all three we derive single-exponential algorithms to compute a minimal witness. Our approaches are based on Farkas certificates for quotient MDPs under probabilistic time-abstracting bisimulations.
The time complexities are relative to the sizes of these quotients, so coarse bisimulations can substantially benefit the approach. 
While comparing two subsystems with respect to their location number or invariance strength is relatively easy, comparing the volume is shown to be \PP-hard. 
This result notably extends also to non-probabilistic timed automata.

An open question is how to extend the scope of witnessing subsystems to probabilistic hybrid automata (PHA).
It is conceivable that our approach extends naturally to \emph{rectangular} PHAs, as they admit finite bisimulation quotients \cite{Sproston11}.
Exploring how PTA subsystems can be used in timed versions of refinement and synthesis algorithms \cite{HermannsWZ2008,CeskaHJK19} is another interesting line of future work.

\bibliographystyle{splncs04}
\bibliography{lit}

\newpage
\appendix

\section{Supplementary material for \Cref{sec:prelims}}\label{app:prelims}

\subsection{Lemmata on DBMs}\label{sub:dbms}

In order to prove the basic properties of canonical DBMs, we need some input on the algebraic structure of DBMs. Denote by $\preceq$ the lexicographic order on $(\mathbb{Z} \cup \{\infty,- \infty\}) \times \{<,\leq\}$ in which $<$ is strictly less than $\leq$. Then $\preceq$ extends to a partial order on DBMs by entrywise comparison, and all subsequent min and max operations refer to this partial order. We define the operations $+,\sqcap,*$ on $(\mathbb{Z} \cup \{\infty,- \infty\}) \times \{<,\leq\}$ as follows~\cite{Dill1990}:
\begin{align*}
(a,\triangleleft_1) + (b,\triangleleft_2) &= (a+b,\min \{\triangleleft_1,\triangleleft_2\}) \\
(a,\triangleleft_1) \sqcap (b, \triangleleft_2) &= \min \{ (a,\triangleleft_1),(b, \triangleleft_2) \} \\
(a,\triangleleft)^* &= \begin{cases}
(0,\leq) & \text{if } (0,\leq) \preceq (a,\triangleleft) \\
(- \infty, <) & \text{otherwise}
\end{cases}
\end{align*}
It is then shown that $(\mathbb{Z} \cup \{\infty,- \infty\}) \times \{<,\leq\}$ with $\sqcap$ as addition and $+$ as multiplication together with the constants $\nb = (\infty,<)$ and $\eb = (0,\leq)$ constitute a regular algebra. Moreover, the set of DBMs $((\mathbb{Z} \cup \{\infty,- \infty\}) \times \{<,\leq\})^{\C \times\C}$ forms a regular algebra where $\sqcap$ is matrix addition and $+$ is matrix multiplication over the scalar operations $\sqcap$ and $+$. 
Then, $M^*$ is defined as $M^0 \sqcap M^1 \sqcap \ldots$, which implies $M^*\preceq M$. Two DBMs $M, N$ with $\Val(M) = \Val(N) \neq \varnothing$ satisfy $M^* = N^*$ (see \cite[Theorem 2]{Dill1990}). Hence, $M^*$ represents the strongest clock constraint with this valuation set and can be seen as the canonical representative DBM for $\Val(M)$. It is a straightforward argument from the projection property of DBMs (see \cite[Lemma 4]{Dill1990}) that for two DBMs $M, N$ with non-empty valuation sets we have $\Val(M) \subseteq \Val(N)$ if and only if $M^* \preceq N^*$.

The following lemma states some basic properties of the canonical DBM, as defined in~\Cref{sub:zone closure}.

\begin{lemma}[Basic properties of the canonical DBM]\label{lem:canonicalDBM}
	Let $R \subseteq \Val(\C)$ be any subset. Then the following hold:
	\begin{enumerate}
		\item $R \subseteq \Val(M_R)$;
		\item $M_R^* = M_R$;
		\item $\Val(M_R)$ is the smallest zone of $\Val(\C)$ that contains $R$;
		\item For any DBM $M$ with $M=M^*$ and $\Val(M)\neq \emptyset$, we have $M = M_{\Val(M)}$.
	\end{enumerate}
\end{lemma}

\begin{proof}
	(1) It is clear from the definition that all points in $R$ satisfy the constraints induced by $M_R$, so we have $R \subseteq \Val(M_R)$. 
	
	(2) Suppose for contradiction that $M_R \neq M_R^*$. Since $M^*\preceq M$ holds for any DBM, we must have a strict inequality $M_R^* \prec M_R$. Hence there exists a pair of indices $i,j$ such that $(M_R^*)_{ij} \prec (M_R)_{ij}$. Suppose that $i=j = 0$ and that no other pair of indices with strict inequality exists, i.e., $(M_R)_{kl} = (M_R^*)_{kl}$ whenever $k\neq0$ or $l\neq0$. Note that 
	\[(M_R^2)_{00} = \bigsqcap_{i=0}^n (M_R)_{0i}+(M_R)_{i0} =  \min_i (M_R)_{0i}+(M_R)_{i0} = (0,\leq) = (M_R)_{00}\]
	which would imply $M_R\preceq M_R^2$ since $(M_R)_{kl} = (M_R^*)_{kl}\preceq (M_R^2)_{kl}$ whenever $k\neq0$ or $l\neq0$. However, by induction we would also have $M_R\preceq M_R^n$ for all $n\geq 2$, so $M_R\preceq M_R^*$ and therefore $M_R = M_R^*$, which is a contradiction. In summary, $M_R$ and $M_R^*$ cannot only differ on $i=j=0$.
	
	Now let $(M^*_R)_{ij} = (b_1,\triangleleft_1)< (b_2,\triangleleft_2) = (M_R)_{ij}$. We first consider the case that $b_1 < b_2$ and the subcase that $i,j \in \C$.
	Take $\epsilon > 0$ small enough such that $b_1+\epsilon < b_2$. By the definition of $M_R$ we have $b_2 = \sup \{p(i) - p(j) \mid p \in R\}$, so there exists $p \in R$ such that $p(i) - p(j) > b_2 - \epsilon = b_1$. This would entail $p\notin\Val(M_R^*) = \Val(M_R)$, which is a contradiction to the aforementioned inclusion $R \subseteq \Val(M_R)$. The subcase where one of the clocks is $0$ is completely analogous.
	
	Finally consider the case that $b_1 = b_2$, $\triangleleft_1 = \; <$ and $\triangleleft_2 = \; \leq$. If $i,j\in\C$, then there must exist $p\in R$ such that $p(i) - p(j) = b_1 = b_2$. But this point will not be contained in $\Val_{\C}(M_R^*)$ due to the strict inequality, which results once more in a contradiction. The case where one of the indices is equal to $0$ is handled similarly. This finishes the proof that $M_R = M_R^*$.
	
	(3) First consider the case that $R$ itself is a zone, so $R = \Val(g)$ for some clock constraint $g$. Let $M_g$ be the associated DBM. One proves along similar lines as in (2) that $M_R\preceq M_g$. This implies that $R \subseteq \Val(M_R) \subseteq \Val(M_g) = \Val(g) = R$, and hence $R = \Val(M_R)$.
	
	For general $R$, let $Z\subseteq\Val(\C)$ be any zone with $R\subseteq Z$. Then $Z = \Val(M_Z)$ for the canonical DBM $M_Z$ of $Z$, as shown in the previous paragraph. From $R\subseteq Z$, we clearly have $M_R\preceq M_Z$ and thus $\Val(M_R) \subseteq \Val(M_Z) = Z$. Therefore, any zone containing $R$ must also contain $\Val(M_R)$.
	
	(4) Let $Z = \Val(M)$. Since $Z$ is a zone, by part (3) we have $\Val(M_Z) = Z = \Val(M)$. It follows then from part (2) that $M_Z = M_Z^* = M^* = M$.
	\qed
\end{proof}

Recall from \cite{BengtssonW04} that the time closure on DBMs is the unary operation $\timeup$ defined by $(\timeup M)_{ij} = M_{ij}$ if $j\neq 0$ and $(\timeup M)_{i0} = (\infty, <)$ otherwise. In words, the time closure removes absolute time bounds on the clocks in $\C$. The next lemma states that the time closure operator is the syntactic analogue of the classical time closure operation on subsets $R\subseteq \Val(\C)$ defined by $\timeup R = \{v+t\in \Val(\C)\mid v\in R \text{ and } t\geq 0 \}$.

\begin{lemma}
  \label{lem:timeup}
	For any DBM $M$ with $M = M^*$ and $\Val(M)\neq \emptyset$, we have $\Val(\timeup M) = \timeup \Val(M)$.
\end{lemma}
\begin{proof}
	We begin with the inclusion $\Val(\timeup M) \subseteq \timeup \Val(M)$. Let $M_{i0}  =(u_i, \triangleleft_i)$ and $M_{0i} = (-l_i, \triangleleft'_i)$ for $0\leq i\leq n$. Any $v\in \Val(\timeup M)$ satisfies all constraints contained in $M$ except possibly for the constraints $M_{i0}  =(u_i, \triangleleft_i)$. As $\Val(M)$ is non-empty, none of the $u_i$ is $-\infty$ and we have $(l_i, \triangleleft'_i)\preceq (u_i, \triangleleft_i)$ for all $i$. If for all $i$ we have $c_i(v)\triangleleft_i u_i$, then we already have $v\in \Val(M)$. If not, let $t = \max_i\{v(i) - u_i\}\geq 0$ and let $a$ be the index attaining this maximum. We assume that $\triangleleft_{a} = \; \leq$, otherwise we add a small $\epsilon$ to $t$. We claim that the valuation $v' = v - t$ lies in $\Val(M)$. The only constraints in $M$ potentially violated by $v'$ are the absolute lower bounds $M_{0 i} = (-l_i, \triangleleft'_i)$. If this was the case, then for some $b$ we would have $v'(b) < l_{b}$. On the other hand $v'(a) = u_{a}$, and thus
	\begin{align*} 
		M^2_{ab} &\preceq M_{a 0} + M_{0 b}\\
						& = (u_{a}, \triangleleft_{a}) + (-l_{b}, \triangleleft'_{b}) \\
						& \prec (v'(a),  \triangleleft_{a}) + (-v'(b),\triangleleft'_{b}) \\
						& =  (v(a)-v(b), \min\{\triangleleft_{a}, \triangleleft'_{b}\}) \\
						& \preceq M_{ab}
	\end{align*}
	where the last inequality follows from $v\in \Val(\timeup M)$. However, this is a contradiction to $M=M^* = M^0 \sqcap M^1 \sqcap M^2 \sqcap ...$.
	
	For the reverse inclusion $\timeup \Val(M)\subseteq\Val(\timeup M)$, let $v\in \timeup \Val(M)$, so there exists $v'\in \Val(M)$ and $t\geq 0$ such that $v = v'+t$. As $v'\in \Val(M)$, the only constraints in $M$ possibly violated by $v$ are those contained in the  column indexed by $c_0$. As these are relaxed to $(\infty, <)$ in $\timeup M$, we have $v\in \Val(\timeup M)$.
	\qed
\end{proof}

\subsection{PTAB preserve reachability probabilities}

\begin{lemma}\label{lem:bisimilar}
	Let $\Scal$ be a TPS and $\sim$ a PTAB on $\Scal$ that respects $\goal$ and $\fail$.
	If $s \sim s'$, then we have $\prb^{*}_{s}(\lozenge \goal) = \prb^{*}_{s'}(\lozenge \goal)$ for $*\in\{\min, \max\}$.
\end{lemma}
\begin{proof}
	We show by induction that for all $i \geq 0$
	\begin{equation} \label{PTAB lemma}
		\prb^{\max}_{s}(\lozenge_{\leq i} \goal) = \prb^{\max}_{s'}(\lozenge_{\leq i} \goal),
	\end{equation}
	where $\lozenge_{\leq i}\goal$ refers to paths reaching $\goal$ in at most $i$ steps (irrespective of their time duration). For $i = 0$ the claim is clear, as both $s$ and $s'$ must be in location $\goal$.
	
	So let $i = i' + 1$.
	For each $(\alpha,\gamma) \in T(s)$, we find $(\alpha,\gamma') \in T(s')$ such that for all $C \in S/_{\sim}$: $\sum_{t \in C}\gamma(t) = \sum_{t \in C}\gamma'(t)$, and vice versa.
	Hence, in particular:
	\begin{align*}
	&\sum_{t \in \supp(\gamma)} \gamma(t) \cdot \prb^{\max}_{t}(\lozenge_{\leq i'} \goal) \\
	= &\sum_{C \in \supp(\gamma/_{\sim})} \left(\sum_{t \in C} \gamma(t) \right) \cdot \prb^{\max}_{t}(\lozenge_{\leq i'} \goal) \\
	= &\sum_{C \in \supp(\gamma'/_{\sim})} \left(\sum_{t \in C} \gamma(t) \right) \cdot \prb^{\max}_{t}(\lozenge_{\leq i'} \goal) \\
	= &\sum_{t \in \supp(\gamma')} \gamma'(t) \cdot \prb^{\max}_{t}(\lozenge_{\leq i'} \goal) \qquad
	\end{align*}
	As
	\[\prb^{\max}_s(\lozenge_{\leq i} \goal) = \sup_{(\alpha,\gamma) \in T(s)} \sum_{t \in \supp(\gamma)} \gamma(t) \cdot \prb^{\max}_{t}(\lozenge_{\leq i'} \goal)\]
	the claim follows. An analogous calculation can be made for $\prb^{\min}$.
	
	To see that \Cref{PTAB lemma} is enough to conclude that $\prb^{\max}_{s}(\lozenge \goal) = \prb^{\max}_{s'}(\lozenge \goal)$ we observe that
	\[\prb^{\max}_{s}(\lozenge \goal) = \lim_{i \to \infty} \prb^{\max}_{s}(\lozenge_{\leq i} \goal)\]
	holds for all $s \in S$ since the same equation is already true for each time-divergent scheduler on $\Scal$.
	\qed
\end{proof}

\bisimreach*
\begin{proof}
	Throughout we write $\M = \M(\Scal/_{\sim})$. The image of a path $\pi$ in $\Scal$ under the quotient map will be denoted by $\overline{\pi}$. We give the prove for $* = \max$, the other case is completely analogous.

\emph{Step 1.} For every memoryless scheduler $\S$ on $\M$ we construct a memoryless scheduler $\S'$ on $\Scal$ such that for all states $s$ of $\Scal$ we have $\Pr^{\S}_{\M, [s]}(\lozenge \goal) = \Pr^{\S'}_{\Scal, s}(\lozenge \goal)$. Since the maximum $\Pr^{\max}_{\M, [s]}(\lozenge \goal)$ is attained already on memoryless schedulers \cite[Lemma 10.102]{BaierK2008}, this suffices to show $\prb^{\max}_{\M, [s]}(\lozenge \goal) \leq \prb^{\max}_{\Scal, s}(\lozenge \goal)$.

Take any scheduler $\S$ on $\M$. The idea is to lift all scheduler decisions of $\S$ along the quotient map to $\Scal$. More precisely, for a state $s$ in $\Scal$ we make a case distinction on whether $\S([s]) = (\tau,\delta_{C})$ or $\S([s]) = (\alpha,\mu/_{\sim})$.

In the first case, we know that for all states $s$ of $\Scal$ there exists $s' \in C$ and $t \in \R^+$ such that $s \overset{t}{\to} s' \in T_{\Scal}(s)$. We set $\S'(s) = (t,\delta_{s'})$. In the other case, we know that there exists $s \overset{\alpha}{\to} \mu\in T_{\Scal}(s)$ such that for all $C \in S/_{\sim}: \sum_{s \in C} \mu(s) = \mu/_{\sim}(C)$. We set $\S'(s) = (\alpha,\mu)$.

Let $\Paths^{\S'}_{s}(\lozenge_{=i} \goal)$ be the set of $\S'$-paths starting in $s$ that reach $\goal$ in exactly $i$ steps.
We show by induction on $i$ that
\[\Pr(\Paths^{\S'}_{s}(\lozenge_{=i} \goal)) = \Pr(\Paths^{\S}_{[s]}(\lozenge_{=i} \goal))\]
holds for all states $s \in S$.
If $i = 0$ then the LHS is $1$ exactly if $s = (\goal,v)$ for some $v$.
Then, by assumption, all states in $[s]$ are in location $\goal$ and hence the RHS is $1$ as well.
Otherwise, both sides of the equation are equal to $0$.

For $i = i' + 1$ we have:
\begin{align*}
\Pr(\Paths^{\S'}_{s}(\lozenge_{=i} \goal)) &= \sum_{s' \in \supp(\mu)} \mu(s') \cdot \Pr(\Paths^{\S'}_{s'}(\lozenge_{=i'} \goal)) \\
&= \sum_{s' \in \supp(\mu)} \mu(s') \cdot \Pr(\Paths^{\S}_{[s']}(\lozenge_{=i'} \goal)) \quad \text{(I.H.)}\\
&= \sum_{[s'] \in \supp(\mu/_{\sim})} \left(\sum_{u \in [s']} \mu(u)\right) \cdot \Pr(\Paths^{\S}_{[s']}(\lozenge_{=i'} \goal)) \\
&= \sum_{[s'] \in \supp(\mu/_{\sim})} \mu/_{\sim}([s']) \cdot \Pr(\Paths^{\S}_{[s']}(\lozenge_{=i'} \goal)) \\
&= \Pr(\Paths^{\S}_{[s]}(\lozenge_{=i} \goal))
\end{align*}
where we assume $\S'(s) = (\alpha,\mu)$ and $\S([s]) = (\alpha,\mu/_{\sim})$.
A similar calculation can be made in the case that $\S'(s) = (t,\delta_{s'})$ and $\S([s]) = (\tau,\delta_{[s']})$.
As $\Pr^{\S}_{[s]}(\lozenge \goal) = \sum_{i\geq 0} \Pr(\Paths^{\S}_{[s]}(\lozenge_{=i} \goal))$ and analogously for $\M$ this finishes the argument for \emph{Step 1.}
\medskip

\emph{Step 2.} We show that given a scheduler $\S$ on $\Scal$ we can find a scheduler $\overline{\S}$ on $\M$ that makes compatible choices on all paths mapping to the same path in $\M$.

As an intermediate step we define a sequence of schedulers $\S_0, \S_1, \S_2, ...$ on $\Scal$ such that: $\S_0 = \S$, $\S_{i}$ and $\S_{i+1}$ do not differ on paths of length at most $i$ and $\Pr_{\Scal, s}^{\S_{i}}(\lozenge\goal) \leq  \Pr_{\Scal, s}^{\S_{i+1}}(\lozenge\goal)$. For the induction step, assume that $\S_{i}$ has been constructed and consider the (infinite-state, finitely-branching) Markov chain $K^{\S_i} = (\Paths_\fin(\Scal), \Pb, s)$ associated to $\Scal$ and $\S_i$, based at some arbitrary state $s\in \Scal$. 
By definition we have $\Pr_{\Scal, s}^{\S_i}(\lozenge\goal) = \Pr_{K^{\S_i}, s}(\lozenge\goal)$.

Define $\S_{i+1}(\pi) = \S_{i}(\pi)$ for every path of length at most $i$. Let $\pi_1, ..., \pi_n$ be all paths of length $i+1$ based at $s$ in $K^{\S_i}$ that map to the same (fixed) path $\overline{\pi}$ in $\M$. These are finitely many as $K^{\S_i}$ is finitely branching. If we write $s_j = \last(\pi_j)$, then in particular, we have $s_j \sim s_k$ for all $1\leq j,k\leq n$. Let $k^*$ be the index of a path among $\pi_1,...,\pi_n$ that attains the maximal value $\max_{1\leq k\leq n}\: \Pr_{K^{\S_i}, \pi_k} (\lozenge\goal)$. We now emulate the subtree in $K^{\S_i}$ based at $\pi_{k^*}$ by a new subtree at $\pi_j$ for all $1\leq j\leq n$.

Formally, let $\mathfrak{A}$ be the scheduler defined on $\Paths_\fin(\Scal,s_{k^*})$ by $\mathfrak{A}(\tau) = \S_i(\pi_{k^*}\tau)$. Then by definition
\[ \Pr^{\mathfrak{A}}_{\Scal, s_{k^*}}(\lozenge\goal) = \Pr_{K^{\S_i}, \pi_{k^*}}(\lozenge\goal)\]
Since $s_j \sim s_{k^*}$ we know by \Cref{lem:bisimilar} that there exists a scheduler $\mathfrak{A}'$ on $\Paths_\fin(\Scal, s_j)$ such that 
\[ \Pr^{\mathfrak{A}}_{\Scal, s_{k^*}}(\lozenge\goal) \leq  \Pr^{\mathfrak{A'}}_{\Scal, s_{j}}(\lozenge\goal)\]
Now we define $\S_{i+1}(\pi_j\tau) = \mathfrak{A}'(\tau)$ for any finite path $\tau$ starting in $s_j$. Then we have
\[ \Pr^{\mathfrak{A'}}_{\Scal, s_{j}}(\lozenge\goal) = \Pr_{K^{\S_{i+1}}, \pi_{j}}(\lozenge\goal)\]
and taking everything together
\[  \Pr_{K^{\S_{i}}, \pi_{j}}(\lozenge\goal)\leq \Pr_{K^{\S_i}, \pi_{k^*}}(\lozenge\goal) \leq \Pr_{K^{\S_{i+1}}, \pi_{j}}(\lozenge\goal). \]
In total we get
\[ \Pr_{\Scal, s}^{\S_i}(\lozenge\goal) = \Pr_{K^{\S_i}, s}(\lozenge\goal) \leq  \Pr_{K^{\S_{i+1}}, s}(\lozenge\goal) = \Pr_{\Scal, s}^{\S_{i+1}}(\lozenge\goal)  \]
as desired. Now we let the scheduler $\S'$ be the limit of the $\S_i$, i.e. for $\pi\in \Paths_\fin(\Scal)$ of length $i$ we let $\S'(\pi) = \S_{i}(\pi)$.
\medskip

\emph{Step 3.} From the scheduler $\S'$ constructed in \emph{Step 2}, we now induce a scheduler on $\M$. By construction, on all finite paths that map to the same path in $\M$, $\S'$ chooses bisimilar actions. Thus the assignment $\overline{\S}(\overline{\pi}) = \overline{\S'(\pi)}$ is well-defined. It is easy to see that then 
\[ \Pr_{\Scal, s}^{\S}(\lozenge\goal) \overset{\emph{Step 2}}{\leq} \Pr_{\Scal, s}^{\S'}(\lozenge\goal) = \Pr_{\M, [s]}^{\overline{\S}}(\lozenge\goal). \]
This implies 
\[\prb_{\Scal, s}^{\max}(\lozenge\goal)\leq \prb_{\M, [s]}^{\max}(\lozenge\goal)\]
and thus finishes the proof.
\qed
\end{proof}

\section{Supplementary material for \Cref{sec:subsystem}}

\begin{example}[\Cref{ex:subsystem} extended]\label{ex:subsystem extended} Recall the PTA $\T$ depicted in \Cref{fig:example_pta}. Note that the principal choice a scheduler in $\T$ has to make is that between $\alpha$ and $\beta$ (and letting time pass accordingly) whenever in $l_1$. The probability to reach $\goal$ depends on the number of times $\alpha$ has been chosen in $l_1$ before the condition $y\leq 2$ appearing in the right-hand branch of $\T$ is violated. The reason for this is that $y$ is never reset and once $y\leq 2$ is violated the right-hand branch does not contribute towards $\Pr^\S_T(\lozenge\goal)$ anymore. This is a consequence of our semantics in which a transition leading to a violation of the invariance condition of the target is automatically redirected to $\fail$.

More precisely, let $\S_{m, n}$ denote any scheduler that selects $\alpha$ in $l_1$ $m$ times while $y\leq 2$ and $n$ times while $y>2$. Since $\alpha$ has guard $x\geq 1$ and reset $x:=0$, one time unit has to pass between any two occurences of $\alpha$. Therefore, the only possible values of $m$ are $0,1$ and $2$. If $m=0$, then $\beta$ is taken immediately and therefore $n=0$. The probability to reach $\goal$ from $(l_0, (x,y))$ when choosing $\alpha$ and without returning to $l_0$ is $p= \frac{3}{5} \cdot\frac{2}{5} +\frac{2}{5} \cdot \frac{1}{2} = \frac{11}{25}$ if $y\leq 2$, and $q =\frac{2}{5} \cdot \frac{1}{2} = \frac{1}{5}$ if $y> 2$. The finite path leading from $l_0$ to $l_0$ and looping through $l_1$ $n$ times has probability $q^n$. Then we can calculate:

\begin{align*}
	 \Pr^{\S_{0,0}}_\T(\lozenge\goal) &= \frac{3}{5} \cdot\frac{2}{5} +\frac{2}{5} \cdot \frac{3}{4} = \frac{27}{50} = 0.54 \\
	 \Pr^{\S_{1,0}}_\T(\lozenge\goal) &= p + q\cdot \frac{27}{50} = \frac{137}{250} = 0.548\\
	 \Pr^{\S_{2,0}}_\T(\lozenge\goal) &= p + q\cdot p + q^2\cdot\frac{27}{50} = \frac{687}{1250} = 0.5496
\end{align*}
and for $n\geq 1$
 \begin{align*}
	 \Pr^{\S_{1, n}}_\T(\lozenge\goal) &= p + q\cdot p + \sum_{k=0}^{n-1}\: q^{2+k}\cdot q + q^{2+n}\cdot \frac{2}{5}\cdot\frac{3}{4}\\
	 \Pr^{\S_{2,n}}_\T(\lozenge\goal) &=  p + q\cdot p + q^2\cdot p+\sum_{k=0}^{n-1}\: q^{3+k}\cdot q + q^{3+n}\cdot \frac{2}{5}\cdot \frac{3}{4}
 \end{align*}

From this description one can derive that $\Pr^{\S_{m,n'}}_\T(\lozenge\goal) \geq \Pr^{\S_{m,n}}_\T(\lozenge\goal)$ holds whenever $n\geq n'\geq 1$. This means that once the right-hand branch of $\T$ does no longer contribute towards the probability of reaching $\goal$, choosing $\beta$ leads to a higher probability. It also follows that $\Pr^{\S_{m,n}}_\T(\lozenge\goal)\geq \Pr^{\S_{0,0}}_\T(\lozenge\goal)$ for $n\geq 1$ and $m$ arbitrary, so $\S_{0,0}$ is a scheduler attaining $\prb^{\min}_\T(\lozenge\goal)$. In order to find a scheduler attaining $\prb^{\max}_\T(\lozenge\goal)$, we explicitly calculate
 \begin{align*}
 	\Pr^{\S_{1,1}}_\T(\lozenge\goal) &= p + q\cdot p + q^{2}\cdot q + q^3\cdot  \frac{2}{5}\cdot\frac{3}{4} = \frac{1346}{2500} = 0.5384\\
	\Pr^{\S_{2,1}}_\T(\lozenge\goal) &= p + q\cdot p + q^2\cdot p +  q^{3}\cdot q + q^{4}\cdot \frac{2}{5}\cdot \frac{3}{4} = \frac{3423}{6250} = 	0.54768
 \end{align*}
which shows that $\S_{2,0}$ is an optimal scheduler.

\end{example}

\subsysandfarkascert*
\begin{proof}
  We first establish some relations between the semantics of $\T$ and $\T'$.
  For this, we denote by $S_{\T}$ the states of $\Scal(\T)$ and by $S_{\T'}$ the states of $\Scal(\T')$.
  \begin{enumerate}[itemsep=2ex]
  \item[(a)] \emph{$\T$ and $\T'$ have the same set of actions, and $S_{\T} \subseteq S_{\T'}$.}

    \emph{Proof:} As $\T$ and $\T'$ have the same set of actions, the actions of the semantics are also the same.
    $S_{\T} \subseteq S_{\T'}$ follows from \ref{item:loccontainment} and \ref{item:invcontainment} of \Cref{def:subsystem}.
  \item[(b)] \emph{For any transition $s \to \mu_\sem'$ (discrete action, or time delay) in $\Scal(\T')$, there exists a transition $s \to \mu_\sem$ in $\Scal(\T)$ such that for all $t \in \supp(\mu_\sem)$ with $t \notin \fail$: $\mu_\sem'(t) \leq\mu_\sem(t)$.}

    \emph{Proof:}
    We first consider discrete transitions.
    Take a transition $(\alpha,\mu_{\sem}')\in T_{\Scal(\T')}(l,v)$ for some state $(l,v)$.
	  There must be $l\overset{g':\alpha}{\longrightarrow}\mu'$ in $\T'$ such that $v \models g'$ and which satisfies the equalities in the definition of the semantics of PTAs.
	  Let $\Phi$ be the injection garantueed to exist by condition \ref{item:weak transition} in~\Cref{def:subsystem}, and let $l\overset{g:\alpha}{\longrightarrow}\mu = \Phi( l\overset{g':\alpha}{\longrightarrow}\mu')$. By (3a) we get $v \models g$, and we therefore have a corresponding transition $(\alpha,\mu_{\sem}) \in T_{\Scal(\T)}(l,v)$.
	  From (3c) in \Cref{def:subsystem} we can conclude that $\mu'(C, l') \in \{\mu(C, l'),0\}$ for all $C\subseteq \C, l'\in\Loc'$ with $l'\neq\fail$. This implies that for states $t$ of $S(\T')$ with $t \neq \fail$ we have $\mu_\sem'(t)\leq\mu_\sem(t)$.
	
	  For a time delay $(t,s') \in T_{\Scal(\T')}(s)$, the same time delay must exist in $\Scal(\T)$.
  \item[(c)] \emph{If $\T'$ is a strong subsystem, then for any transition $s \to \mu_\sem$ (discrete action, or time delay) in $\Scal(\T)$ such that $s \in S_{\T'}$, there exists a transition $s \to \mu_\sem'$ in $\Scal(\T')$ such that for all $s' \in \supp(\mu_\sem)$ with $s' \notin \fail$: $\mu_\sem'(s') \leq \mu_\sem(s')$.}
    
    \emph{Proof:}
	  We again first consider discrete actions, so take $(\alpha,\mu_{\sem}) \in T_{\Scal(\T)}(l,v)$.
    Then there exists a corresponding transition $l\overset{g:\alpha}{\longrightarrow}\mu$ in $\T$, and in particular $v\models g$.
	  We use (3$^*$) of \Cref{def:subsystem} to get $l\overset{g':\alpha}{\longrightarrow}\mu' = \Psi(l\overset{g:\alpha}{\longrightarrow}\mu)$.
    This is a transition in $\T'$ such that $g' \equiv g \land \inv'(l)$ by (3a$^*$).
	  Now, from $v \models g$ and $v \models \inv'(l)$ we can derive $v \models g'$.
	  Hence, there exists a transition $(\alpha,\mu'_{\sem}) \in T_{\Scal(\T')}(l,v)$.
	The required relation between $\mu_{\sem}$ and $\mu_{\sem}'$ follows in the same way as in (b).
	
  Now take a time delay $(t,\delta_{(l,v{+}t)}) \in T_{\Scal(\T)}(l,v)$ where $(l,v) \in S_{\T'}$.
	Then we have $v\models\inv'(l)$ and since $(l, v{+}t)\in S_{\T}$ we have $v{+}t\models \inv(l)$.
  By condition (4) of~\Cref{def:subsystem} it follows that $v {+} t \models \inv'(l)$ and hence $(l,v{+}t) \in S_{\T'}$. Therefore the transition $(t,\delta_{(l,v{+}t)})$ lies also in $T_{\Scal(\T')}(l,v)$, which completes the proof.
  \end{enumerate}
  
  Now let $\M_{S'}$ be the MDP-subsystem of $\M$ induced by $S'$, as defined in~\cite[Notation 5.3]{FunkeJB20}, which essentially deletes from $\M$ all states not contained in $S'$ and redirects edges to states outside of $S'$ to $\fail$. 
  To show the main claim, we want to establish the following chain of inequalities:
  \begin{equation}\label{eq:max theorem}
  	\prb_{\T'}^{\max}(\lozenge \goal) \leq \prb_{\M_{S'}}^{\max}(\lozenge \goal) \leq \prb_{\M}^{\max}(\lozenge \goal)
  \end{equation}
  and, if $\T'$ is a strong subsystem:
  \begin{equation}\label{eq:min theorem}
  \prb_{\T'}^{\min}(\lozenge \goal) \leq \prb_{\M_{S'}}^{\min}(\lozenge \goal) \leq \prb_{\M}^{\min}(\lozenge \goal)
  \end{equation}
  In both cases, the second inequality follows from \cite[Lemma 4.4]{FunkeJB20}.
  
  For the first inequality we let $\Scal_{S'}$ be the TPS that includes exactly the states of $\Scal(\T)$ whose equivalence class lies in $S'$.
  More precisely, let
  \[\Scal_{S'} = \left(\bigcup_{[s] \in S'} [s], \Act\uplus\:\R_+ , T_{S'}, s_0\right),\] 
  where the transitions in $T_{S'}$ correspond exactly to the transitions of $\Scal(\T)$ for the given state, with the exception that successor states that are not present in $\Scal_{S'}$ are replaced by $\fail$.
  With this definition in place, we aim to show
  \[\prb_{\T'}^{*}(\lozenge \goal) \leq \prb_{\Scal_{S'}}^{*}(\lozenge \goal) = \prb_{\M_{S'}}^{*}(\lozenge \goal) \]

  As $\Scal_{S'}$ is the TPS that merges all states that are not in $S'$ with $\fail$, and elements of $S'$ are complete equivalence classes under $\sim$, the restriction of $\sim$ to $\bigcup_{[s] \in S'} [s]$ is a PTAB on $\Scal_{S'}$ that respects $\goal$ and $\fail$.
  Furthermore, the corresponding quotient is $\M_{S'}$.
  Now $\prb_{\Scal_{S'}}^{*}(\lozenge \goal) = \prb_{\M_{S'}}^{*}(\lozenge \goal)$ follows by~\Cref{lem:bisimreach} for both $\min$ and $\max$ reachabiliy probabilities.
  \medskip

  We now consider max-probabilities, and show $\prb_{\T'}^{\max}(\lozenge \goal) \leq \prb_{\Scal_{S'}}^{\max}(\lozenge \goal)$.
  It is enough to show that for every scheduler $\S$ for $\Scal(\T')$ there exists a scheduler $\S'$ for $\Scal_{S'}$ such that $\Pr^{\S}_{\Scal(\T')}(\lozenge \goal) \leq \Pr^{\S'}_{\Scal_{S'}}(\lozenge \goal)$.
 In order to prove this, take a scheduler $\S$ for $\Scal(\T')$ and define $\S'$ by mimicking $\S$ on paths that exists in $\Scal(\T')$, and arbitrarily otherwise.
  This is possible by (a) and (b), as proven above, and it also directly follows by (b) that $\Pr^{\S}_{\Scal(\T')}(\lozenge \goal) \leq \Pr^{\S'}_{\Scal_{S'}}(\lozenge \goal)$.
  
  Next, we consider min-probabilities, where we need to assume that $\T'$ is a strong subsystem and show $\prb_{\T'}^{\min}(\lozenge \goal) \leq \prb_{\Scal_{S'}}^{\min}(\lozenge \goal)$.
  Here it suffices to show that for every scheduler $\S'$ for $\Scal_{S'}$ there exists a scheduler $\S$ for $\Scal(\T')$ such that $\Pr^{\S}_{\Scal(\T')}(\lozenge \goal) \leq \Pr^{\S'}_{\Scal_{S'}}(\lozenge \goal)$.
	Let $\S'$ be such a scheduler for $\Scal_{S'}$ and define a scheduler $\S$ for $\Scal_{\T'}$ by mimicking $\S'$ on every path.
  This is possible by (a) and (c) from above, and again (c) directly implies that $\Pr^{\S}_{\Scal(\T')}(\lozenge \goal) \leq \Pr^{\S'}_{\Scal_{S'}}(\lozenge \goal)$.
  This completes the proof of \Cref{eq:max theorem} and \Cref{eq:min theorem}.
  
  It follows that $\M_{S'}$ is a witnessing MDP subsystem in the sense of \cite[Definition 4.1]{FunkeJB20}. Furthermore, by \cite[Theorem 5.4]{FunkeJB20} we can also find the corresponding Farkas certificates supported on the staates of $\M_{S'}$, i.e., on $S'$.
	\qed
\end{proof}

\zoneclosure*
\begin{proof}
	(1) $\Val(M \sqcup N)$ obviously contains $R := \Val(M) \cup \Val(N)$.
	In view of \Cref{lem:canonicalDBM}, part (4), we have $M = M_{\Val(M)}$ and $N = M_{\Val(N)}$, and thus $M\preceq M_R$ and $N\preceq M_R$. Therefore $M\sqcup N \preceq M_R$. Now the claim follows from \Cref{lem:canonicalDBM}, part (3).
	
	(2) Assume, for contradiction, that  $(M \sqcup N)^* \prec (M \sqcup N)$.
	Then, there exist $i,j$ such that $(M \sqcup N)^*_{ij} \prec (M \sqcup N)_{ij} = \max \{M_{ij}, N_{ij}\}$.
	Let $(M \sqcup N)^*_{ij} = (a,\triangleleft_1)$ and assume, w.l.o.g., that $\max \{ M_{ij}, N_{ij}\} = M_{ij} = (b,\triangleleft_2)$.
	We make the following case distinction:
	\begin{enumerate}
		\item[(i)] Assume that $a < b$ holds.
		There is no point $p \in \Val(M \sqcup N) =  \Val((M \sqcup N)^*)$ such that $p(i) - p(j) > a$.
		On the other hand, we deduce from $M = M^* = M_{\Val(M)}$ (see \Cref{lem:canonicalDBM}, part (4)) that there exist points in $\Val(M)$ such that either $p(i) - p(j) = b$ (if $\triangleleft_2 = \; \leq$) or $p(i) - p(j)$ is arbitrarily close to $b$ (if $\triangleleft_2 = \; <$).
		Both cases yield a contradiction to $\Val(M) \subseteq \Val(M \sqcup N)$.
		\item[(ii)] Assume that $a = b$, $\triangleleft_1 = \; <$ and $\triangleleft_2 = \; \leq$ hold.
		Again, as $M = M^*=M_{\Val(M)}$ there exists a points $p \in \Val(M)$ such that $p(i) - p(j) = b$, but this point is not contained in $\Val(M \sqcup N)$ due to $(M\sqcup N)^*_{ij} = (b, <)$.
	\end{enumerate}
	\qed
\end{proof}

\regioncup*
\begin{proof}
 We have $R_i \subseteq \Val(M_{R_i})$ and $M_{R_i} =M_{R_i}^*$ by~\Cref{lem:canonicalDBM}.
  The claim now follows by inductive application of~\Cref{lem:maxunion}.\qed
\end{proof}

\rtoptasubsset*
\begin{proof}
		We show that $\T_R^w$ satisfies the conditions \ref{item:loccontainment}-\ref{item:weak transition} from~\Cref{def:subsystem} and $\T_R^s$ additionally satisfies (3$^*$) and (4).
		Condition \ref{item:loccontainment} is trivially true.
		
		Condition \ref{item:invcontainment} requires that for all $l \in \Loc'$ we have $\inv'(l) \Implies \inv(l)$. We first show this for $\inv^w(l) = M_l^w = \bigsqcup_{s\in R} M_{s_{|l}}$.
		From \Cref{lem:regioncup} it follows that $\Val(\inv^w(l))$ is the smallest zone that contains $\bigcup_{s\in R}\; s_{|l}$. Since this set lies in the zone $\Val(\inv(l))$, we have $\Val(\inv^w(l))\subseteq \Val(\inv(l))$ and hence by definition $\inv^w(l)\Implies \inv(l)$. For $\inv^s(l) = M_l^s = (\timeup M_l^w) \sqcap M_{\inv(l)}$, the property $\inv^w(l)\Implies \inv(l)$ is trivial. The remaining conditions \ref{item:weak transition} for $\T^w$ and (3$^*$) and (4) for $T_R^s$ follow immediately from the construction.
		\qed
\end{proof}

\witsubsysandfarkcert*
\begin{proof}
	Consider the MDP subsystem $\M_R$ of $\M$ as defined in \cite[Notation 5.3]{FunkeJB20}, which essentially deletes from $\M$ all states not contained in $R$ and redirects edges to states outside of $R$ to $\fail$. 
  Then \cite[Theorem 5.4]{FunkeJB20} states that if there exists a Farkas certificate $\zb \in \P^{\min}_{\M}(\lambda)$ with $\supp(\zb) \subseteq R$, then $\M_R$ is a witness for $\prb_{\M, s_0}^{\min}(\lozenge\goal)\geq\lambda$, i.e. $\prb_{\M_R, s_0}^{\min}(\lozenge\goal)\geq\lambda$. 
	
	We now wish to show that $\T_R^s$ is a witness for $\prb_{\T, l_0}^{\min}(\lozenge\goal)\geq\lambda$ by establishing the chain of inequalities
	\begin{equation}\label{eq:Farkas to witness}
	  \prb_{\T_R^s}^{\min}(\lozenge\goal) =  \prb_{\Scal(\T_R^s)}^{\min}(\lozenge\goal) \geq \prb_{\Scal_R}^{\min}(\lozenge\goal) = \prb_{\M_R}^{\min}(\lozenge\goal)\geq \lambda,
	\end{equation}
	where $\Scal_R$ is the TPS that includes exactly the states of $\Scal(\T)$ whose equivalence class lies in $R$ (compare also the proof of \Cref{prop:subsysandfarkascert}). More precisely, let
	$\Scal_{R} = \left(\bigcup_{[s] \in R} [s], \Act\uplus\:\R_+ , T_{R}, s_0\right)$,
	where the transitions in $T_{R}$ correspond exactly to the transitions of $\Scal(\T)$ for the given state, with the exception that successor states that are not present in $\Scal_{R}$ are replaced by $\fail$. Then the quotient of $\Scal_R$ under the restriction of $\sim$ is precisely $\M_{R}$, and as a consequence $\prb_{\Scal_R}^{\min}(\lozenge\goal) = \prb_{\M_R}^{\min}(\lozenge\goal)$ by~\Cref{lem:bisimreach} for both $\min$ and $\max$ reachabiliy probabilities. As the first equality in (\ref{eq:Farkas to witness}) follows from the definition and the final inequality in (\ref{eq:Farkas to witness}) has been derived in the first paragraph of this proof, we are left to show that 
	\begin{equation}
	\label{eq:Farkas to witness part} 
	 \prb_{\Scal(\T_R^s)}^{\min}(\lozenge\goal) \geq \prb_{\Scal_R}^{\min}(\lozenge\goal)
	 \end{equation}
	
	Take a state $(l,v)$ of $\Scal_R$. This means that $[(l,v)]\in R$ and thus $l\in\Loc(\T_R^s)$ as $[(l,v)]_{|l}\neq\emptyset$. Moreover since $\inv^w(l) = \bigsqcup_{s \in R} M_{s_{|l}}$, we have $v\models\inv^w(l)$ and therefore also $v\models \inv^s(l) = \timeup(\inv^w(l)) \sqcap M_{\inv(l)}$. Hence, $(l,v)$ is a state of $\Scal(\T_R^s)$.
	
	Next let $s = (l,v)$ be a state of $\Scal_R$ and let $s\overset{\alpha}{\longrightarrow}\mu'_\sem \in T_{\Scal(\T_R^s)}(s)$. This transition comes from a transition $l\overset{g^s:\alpha}{\longrightarrow}\mu' \in T_{\T_R^s}(l)$ satisfying the equations appearing in the definition of PTA semantics. By the definition of $\T_R^s$, there exists $l\overset{g:\alpha}{\longrightarrow}\mu \in T_{\T}(l)$ such that $\mu'(C, l') = \mu(C,l')$ whenever $[(l',v[C:=0])] \in R$. This induces a transition $s\overset{\alpha}{\longrightarrow}\mu_\sem \in T_{\Scal(\T)}(s)$ and accordingly a transition $s\overset{\alpha}{\longrightarrow}\overline{\mu}_\sem \in T_{\Scal_R}(s)$ with $\overline{\mu}_\sem(t) = \mu_\sem(t) = \mu'_\sem(t)$ for all states $t$ of $\Scal_R$.
 In summary, every transition of $\Scal(\T_R^s)$ based at a state in $\Scal_R$ is mirrored by a transition in $\Scal_R$ with the same distribution on states in $\Scal_R$ and remaining probability redirected to $\fail$. Completely analogous reasoning shows, vice versa, that every path in $\Scal_R$ is also a path in $\Scal(\T_R^s)$.
	
	In order to prove (\ref{eq:Farkas to witness part}) we need to argue that for every scheduler $\S$ on 
	$\Scal(\T_R^s)$ there exists a scheduler $\S'$ on $\Scal_R$ with  $\Pr_{\Scal(\T_R^s)}^{\S}(\lozenge\goal) \geq \Pr_{\Scal_R}^{\S'}(\lozenge\goal)$. With the notation of the previous paragraph, we define $\S'(\pi) = s\overset{\alpha}{\longrightarrow}\overline{\mu}_\sem$ if $\S(\pi) = s\overset{\alpha}{\longrightarrow}\mu'_\sem$ for every finite path $\pi$ in $\Scal_R$. Since $\overline{\mu}_\sem$ coincides with $\mu'_\sem$ on the states of $\Scal_R$ and redirects the remaining probability to $\fail$, the desired inequality $\Pr_{\Scal(\T_R^s)}^{\S}(\lozenge\goal) \geq \Pr_{\Scal_R}^{\S'}(\lozenge\goal)$ follows.

	The statement about $\T_R^w$ is completely analogous.
	\qed
\end{proof}

\section{Supplementary material for \Cref{sec:minimal}}

\begin{example}[\Cref{ex:minimal} extended]\label{ex:minimal extended}
	In this example we elaborate on \Cref{tab:overview} which lists loc-, inv-, and vol-minimal witnesses for the PTA of \Cref{fig:example_pta}. 
	
	For loc-minimal witnesses we only list the  locations remaining in the subsystem. As the invariance does not affect loc-minimality, one may or may not shrink the invariance of these locations as long as the required probabilistic threshold is kept. The fact that keeping $l_0$ and $l_1$ induces a loc-minimal subsystem $\T_1$ for $\prb^{\max}$ follows from the computation
	\begin{equation}
		\prb^{\max}_{\T_1}(\lozenge\goal) = \sum_{k\geq 1}^\infty \: \left(\frac{2}{5}\cdot\frac{1}{2}\right)^k = \frac{1}{4} \geq \frac{6}{25}
	\end{equation}
	On the other hand, $\T_1$ does not provide a witness for $\prb^{\min}_{\T}(\lozenge\goal) \geq \frac{6}{25}$ since the entire action $\beta$ is redirected to $\fail$, and hence $\prb^{\min}_{\T_1}(\lozenge\goal) = 0$. The fact that keeping $l_0$ and $l_2$ induces a loc-minimal subsystem $\T_2$ for both $\prb^{\max}_\T(\lozenge\goal) \geq \frac{6}{25}$ and $\prb^{\min}_\T(\lozenge\goal) \geq \frac{6}{25}$ is easy to see; in both cases for $*\in \{\min,\max\}$ we have $\prb^{*}_{\T_2}(\lozenge\goal) = \frac{2}{5}\cdot\frac{3}{5}=\frac{6}{25}$. 
	
	We now turn to the inv-minimal subsystems displayed in \Cref{tab:overview} and begin with $\prb^{\max}$. The first one (containing $l_0$ and $l_1$ with shrunk invariants), say $\T'_1$, encodes the first two runs through the subsystem $\T_1$ from above, and thus 
	\[\prb^{\max}_{\T'_1}(\lozenge\goal) =  \frac{2}{5}\cdot\frac{1}{2} + \left(\frac{2}{5}\cdot\frac{1}{2}\right)^2 = \frac{6}{25}\]
	It is inv-minimal precisely because the first run (i.e. the first summand) does not suffice to satisfy the threshold $\frac{6}{25}$. The second subsystem (containing $l_0$ and $l_2$ with shrunk invariants), say $\T'_2$, encodes the fact that one has to wait for one time unit in $l_2$ before the guard $y\geq 1$ of the only action at $l_2$ is satisfied. The third and fourth subsystem, say $\T_{13}$ and $\T_{13}'$, have -- in contrast to $\T'_1$ -- smaller invariants of $l_0$ and $l_1$ which prevent visiting $l_0$ a second time. However, by adding $l_3$, the action $\beta$ becomes available, thus leading to 
	\[ \prb^{\max}_{\T_{13}}(\lozenge\goal) = \prb^{\max}_{\T'_{13}}(\lozenge\goal) = \frac{2}{5}\cdot \frac{3}{4} = \frac{3}{10}\geq \frac{6}{25}.\]
	The only difference between $\T_{13}$ and $\T_{13}'$ is the location in which the time is spent that one needs to wait before the guard of $\beta$ is satisfied. Since strong subsystems must be closed under time successors lying in the original invariance condition (see (4) of \Cref{def:subsystem}), none of the inv-minimal subsystems for $\prb^{\max}$ are strong.
	
	We now turn to the two mentioned inv-minimal subsystems for  $\prb^{\min}$. The first one depicted (consisting of $l_0$ and $l_2$ with shrunk invariants) represents the right-hand side of the PTA. The invariance shown for $l_2$ is the smallest possible invariance which contains $(0,0)$ and is closed under time successors in $\inv_\T(l_2)$.
  No matter how much time is spent in $l_2$ before taking the only action in that location, the probability to reach $\goal$ is $\frac{6}{25}$.
	
	As to the second inv-minimal subsystem $\prb^{\min}$, the above discussion shows that if $l_1$ is in the subsystem, keeping $l_3$ is necessary to avoid $\prb^{\min}(\lozenge \goal) = 0$. The invariances of this subsystem are chosen in a way such that (a) no matter how much time one stays in $l_1$ in the first run of the system, one can pass through $l_0$ again so that probability from the second run is added in order to reach the threshold $\frac{6}{25}$, and (b) whenever the guard $x\leq 1$ of $\beta$ is satisfied, then taking $\beta$ actually leads to $l_3$ (that is, the new invariant of $l_3$ is satisfied after taking the transition).
\end{example}

\begin{lemma}
  Deciding $\prb^{\max}_{\T}(\lozenge \goal) \geq 1$ ($\prb^{\min}_{\T}(\lozenge \goal) \geq 1$) stays \EXPTIME-hard (\PSPACE-hard) under the assumption that all time-divergent schedulers reach $\goal$ or $\fail$ with probability one.
\end{lemma}
\begin{proof}
    In~\cite{LaroussinieS07} is is shown that deciding $\prb_{\T}^{\max}(\lozenge \goal) \geq 1$ is \EXPTIME-hard.
  The proof goes by a direct reduction from the non-emptiness problem of a linearly bounded, alternating Turing machine (Theorem 3.1).
  It is also noted that one can assume without loss of generality that no configuration of the Turing machine is repeated in any run.
  This can always be enforced by letting a counter (encoded in binary) run along the computation, which is increased at every step until the maximal number of possible configurations.
  As the configurations of the Turing machine are encoded in the clock valuation of the PTA, the construction of Theorem 3.1 for such Turing machines yield PTA in which no state can be repeated on any path.
  Furthermore, as the number of configurations of the TM is finite, each time-divergent path will eventually reach $\goal$ (the accepting configuration) or $\fail$ (when the counter exceeds the maximum bound).
  It follows that the problem is already hard for PTA under the mentioned assumption.

  The problem $\prb_{\T}^{\max}(\lozenge \goal) > 0$ is \PSPACE-hard as it essentially asks for \emph{any} path that reaches $\goal$, and hence can be used to encode non-probabilistic reachability problem, which was shown to be \PSPACE-hard in~\cite[Theorem 4.17]{AlurD1994}.
  Again, this proof goas via a reduction from a linearly bounded Turing machine and by a similar argument as before it can be seen that one can assume that all time-divergent paths reach $\goal$ or $\fail$.
 
  Under these assumptions, $\prb_{\T}^{\max}(\lozenge \goal) > 0$ can be reduced to $\prb_{\T}^{\min}(\lozenge \goal) \geq 1$ by replacing $\goal$ and $\fail$, and hence it follows that this problem is \PSPACE-hard.
\end{proof}

\partialorders*
\begin{proof}
	Let $\T_1, \T_2$ be PTAs satisfying $\T_1 \leq_{\inv} \T_2$.
	Then $\T_1 \leq_{\loc} \T_2$ follows directly by $\Loc(\T_1) \subseteq \Loc(\T_2)$ and $\T_1 \leq_{\vol} \T_2$ follows by $\inv_{\T_1}(l) \Implies \inv_{\T_2}(l)$ for all $l \in \Loc$.
	
	By considering two PTAs with a single location and different invariants, it becomes clear that $\T_1 \leq_{\loc} \T_2$ does not imply $\T_1 \leq_{\vol} \T_2$ nor $\T_1 \leq_{\inv} \T_2$.
	To see that $\T_1 \leq_{\vol} \T_2$ does not imply $\T_1 \leq_{\loc} \T_2$ or $\T_1 \leq_{\inv} \T_2$ in general it suffices to arrange $\T_1$ to have one location more than $\T_2$, but less volume in total.
	\qed
\end{proof}

\locminimal*
\begin{proof}
	``$\implies$'': Let $\T'$ be a strong subsystem of $\T$ such that $\prb^{\min}_{\T'}(\lozenge \goal)
	\geq\lambda$ with at most $k$ locations.
  Let $S' = \{[s]\in S\mid s \:\text{ is a state in }\: \Scal(\T')\}$.
	Then, by \Cref{prop:subsysandfarkascert} there exists a Farkas certificate $\zb$ for $\prb^{\min}_{\M}(\lozenge \goal) \geq \prb^{\min}_{\T'}(\lozenge \goal)$ (and hence for $\prb^{\min}_{\M}(\lozenge \goal) \geq \lambda$) satisfying $\supp(\zb) \subseteq S'$.
	Let $\zeta$ be defined by
	\[
	\zeta_l =
	\begin{cases}
	1 & \text{ if there exists a } v \in \Val(\C) \text{ s.t. } \zb_{[(l,v)]} > 0 \\
	0 & \text{ otherwise }
	\end{cases}
	\]
	Then, $(\zb,\zeta)$ satisfies (\ref{eq:loc constr}).
  Here we use that if $\zb \in \P_{\M}^{\min}(\lambda)$, then $\zb_{s} \leq 1$ holds for all $s \in S$ (see~\cite[Lemma 3.1.]{FunkeJB20}).
	Also, $\zeta$ has at most $k$ non-trivial entries as $S'$ contains states from at most $k$ different locations (this uses the fact that $\sim$ distinguishes locations) and $\supp(\zb) \subseteq S'$.
	
	``$\Longleftarrow$'': Let $(\zb,\zeta)$ be a solution of (\ref{eq:loc constr}) such that $\zeta$ has at most $k$ non-trivial entries.
	By~\Cref{prop:witsubsysandfarkcert} it follows that $\T_{\supp(\zb)}^s$ is a witness for $\prb^{\min}_\T(\lozenge \goal) \geq \lambda$.
	The locations of $\T_{\supp(\zb)}^s$ are $\Loc' = \{l \in \Loc \mid \exists v. \; [(l,v)] \in  \supp(\zb) \} \cup \{\goal,\fail\}$, and as $\zeta$ is non-trivial in at most $k$ entries, it follows that $|\Loc' \setminus \{\goal, \fail\}| \leq k$.
	\qed
\end{proof}

\locminimalmilp*
\begin{proof}
  It is enough to show that (\ref{eq:loc milp}) can be solved in time $2^{|\Loc|} \cdot \poly(|\M|)$.
  This can be done by enumerating the vectors $\vb \in \{0,1\}^{|\Loc|}$ and checking for each of them whether a $\zb$ exists such that $(\zb,\vb)$ satisfies (\ref{eq:loc constr}).
  This check amounts to solving a linear program of size $|\Loc| + |\M|$.
  Finally, a vector $\vb$ with a maximal amount of zeros is returned, and it encodes a loc-minimal witness by \Cref{prop:locminimalws}.
  \qed
\end{proof}

\invminimal*
\begin{proof}	
	From~\Cref{prop:witsubsysandfarkcert} it follows that $\T_{\supp(\zb)}^s$ is a witness for $\prb^{\min}_{\T}(\lozenge \goal) \geq \lambda$.
	Assume that it is not $\leq_{\inv}$-minimal, that is, there exists a witness $\T'$ for $\prb^{\min}_{\T}(\lozenge \goal) \geq \lambda$ such that $\T' <_{\inv} \T_{\supp(\zb)}^s$.
	Let $\M$ be the $\sim\,$-quotient of of $\T$ with states $S \cup \{\goal,\fail\}$ and let $S' = \{[s] \in S \mid s \text{ is a state of } \Scal(\T')\}$.
	By~\Cref{prop:subsysandfarkascert} there is a Farkas certificate $\zb' \in \P^{\min}_{\M}(\lambda)$ with $\supp(\zb') \subseteq S'$. We now define a vector $\vb$ that will be the second component in a solution of (\ref{eq:inv milp}).
	First, entries of $\vb$ that refer to locations not in $\T'$ are set to $0$.
	For all other locations $l$, $c_i,c_j \in \C$ with $j \neq 0$, and $k \in \{-2K,\ldots,2K\}$, let $(M_{\inv_{\T'}(l)})_{ij} = (a,\triangleleft)$. We define
	\[ {\vb}_{ij}^{l}(k) = 
	\begin{cases}
	1 & \text{if } a > \ceil*{k/2} \\
	1 & \text{if } a = \ceil*{k/2} \text{ and } \triangleleft = \, \leq \\
	1 & \text{if } a = \ceil*{k/2}, k \text{ is odd, and } \triangleleft = \, < \\
	0 & \text{otherwise}
	\end{cases}
	\]
	
	We now argue that $(\zb',\vb)$ satisfies (\ref{eq:inv constr}). The condition $\zb' \in \P^{\min}_{\M}(\lambda)$ therein holds by assumption, and the condition $\vb_{ij}^l(n) \leq \vb_{ij}^l(n{-}1)$ is immediate. Now take $[(l,v)] \in S$ with $\zb'_{[(l,v)]} > 0$ (for the other states in $\M$ there is nothing to show). From $\supp(\zb')\subseteq S'$, it follows that $[(l,v)] \in S'$ and hence there exists a $(l,v') \sim (l,v)$ (using the assumption that $\sim$ distinguishes locations) such that $v' \models \inv_{\T'}(l)$. As, by assumption, $\sim$ distinguishes in each location valuations which are distinguishable by clock constraints, we have $v'' \models \inv_{\T'}(l)$ for all $(l,v'') \in [(l,v)]$.\footnote{This argument uses the fact that there is an upper-bound $K$ on all clocks, as the standard region construction would not differentiate between valuations exceeding the greatest appearing integer in any clock constraint.}
 As a consequence, we get $(M_{[(l,v)]})_{ij} \preceq (M_{\inv_{\T'}(l)})_{ij} = (a,\triangleleft)$ for all $c_i,c_j \in \C$\footnote{$M_{[(l,v)]}$ is defined as $M_{s|_l}$ in \Cref{def:indsubsys}. As $\sim\,$ distinguishes locations, we omit the $|_l$ subscripts.}.
	
	Now we distinguish the following three cases corresponding to the case distinction in (\ref{eq:inv constr}):
	\begin{enumerate}
		\item If $(M_{([l,v])})_{ij} = (b,<)$ for some $b\in \mathbb{Z}$, we need to check that $\vb_{ij}^l(2b{-}1) = 1$.
		As $(b,<) \preceq (a,\triangleleft)$ we have either $b < a$, or $b = a$ and $\triangleleft$ might be $<$ or $\leq$.
		In the first case we have $a > b= \ceil*{(2b{-}1)/2}$ and hence $\vb_{ij}^l(2b{-}1) = 1$.
		In the second, we have $a = b= \ceil*{(2b{-}1)/2}$. By inspecting the definition of $\vb$ on odd values, we see that $\vb_{ij}^l(2b{-}1) = 1$, irrespective of the value of $\triangleleft$.
		
		\item If $(M_{([l,v])})_{ij} = (b,\leq)$ for some $b\in \mathbb{Z}$, we need to check that $\vb_{ij}^l(2b) = 1$.
		As $(b,\leq) \preceq (a,\triangleleft)$ we have either $b <  a$ or $b=a$ and $\triangleleft = \leq$. By inspecting the definition of $\vb$ one sees that in both cases we have $\vb_{ij}^l(2b) = 1$.
	\end{enumerate}
	
	We conclude that $(\zb',\vb)$ satisfies (\ref{eq:inv constr}). Now we argue that 
	\begin{equation}\label{eq:inv minimal} 
		\sum_{l,i,j,k}\vb_{ij}^l(k) < \sum_{l,i,j,k}\xi_{ij}^l(k)
	\end{equation}
	which would contradict the fact that $(\zb,\xi)$ is optimal.
	We first show that the LHS in \cref{eq:inv minimal} is less or equal than the RHS. We establish this fact summand-wise.
	
	Let us assume first that $k$ is odd and fix some $l \in \Loc$, $c_i,c_j \in \C$ such that $c_j \neq \zerocl$.
	Then, if $\vb_{ij}^l(k) = 1$, we have $(\ceil*{k/2},<) \preceq (M_{\inv_{\T'}(l)})_{ij}$.
	As $\T' <_{\inv} \T_{\supp(\zb)}^s$ we have $(M_{\inv_{\T'}(l)})_{ij} \preceq (M_l^s)_{ij}$, where $M_l^s$ is the invariant DBM for $\T_{\supp(\zb)}^s$.
	By the construction of $M_l^s$ (see~\Cref{def:indsubsys}) there exists some $[(l,v)] \in S$ such that $\zb_{[(l,v)]} > 0$ and $(M_{[(l,v)]})_{ij} = (M_l^s)_{ij}$.
	So, $(\ceil*{k/2},<) \preceq (M_{[(l,v)]})_{ij} = (b, \triangleleft)$ holds.
	As $k$ is odd, $b \geq \ceil*{k/2}$ implies $2b{-}1 \geq k$.
	Also, from $\zb_{[(l,v)]} > 0$ it follows that $\xi_{ij}^l(2b{-}1) = 1$.
	This, with $\xi_{ij}^l(n-1) \geq \xi_{ij}^l(n)$, yields $\xi_{ij}^l(k) = 1$.
	The case where $k$ is even is similar.
	
	It remains to show that the LHS in \cref{eq:inv minimal} is strictly smaller than the RHS.
	As $\T' <_{\inv} \T_{\supp(\zb)}^s$, either the locations of $\T'$ are strictly included in the locations of $\T_{\supp(\zb)}^s$, or, for some location $l$ the invariant in $\T'$ is strictly stronger than the invariant of $\T_{\supp(\zb)}^s$.
	In the first case there is some location $l$ such that some $\xi_{ij}^l(k)$ is $1$, whereas no $\vb_{ij}^l(k)$ is $1$, which yields the claim.
	
	In the other case, there is some location $l$, and $c_i,c_j \in \C$ with $c_j \neq \zerocl$ such that $(M_{\inv_{\T'}(l)})_{ij} \prec (M_l^s)_{ij}$.
	The reason that we can exclude $c_j = \zerocl$ is that both $\T'$ and $\T_{\supp(\zb)}^s$ are strong subsystems of $\T$ and hence need to agree with $\T$ on all time-upper bounds of individual clocks (see condition (4) of~\Cref{def:subsystem}).
	Let $(M_{\inv_{\T'}(l)})_{ij} = (a,\triangleleft_1)$ and $(M_l^s)_{ij} = (b,\triangleleft_2)$.
	Again, there is some $[(l,v)] \in S$ such that $\zb_{[(l,v)]} > 0$ and $(M_{[(l,v)]})_{ij} = (b,\triangleleft_2)$.
	
	First, consider the case $a < b$.
	We have $\xi_{ij}^l(2b{-}1) = 1$.
	As $a < b = \ceil*{(2b{-}1)/2}$ we get $\vb_{ij}^l(2b{-}1) = 0$. Secondly, assume that $a = b$, $\triangleleft_1 = \, <$ and $\triangleleft_2 = \, \leq$.
	We have $\xi_{ij}^l(2a) = 1$ but as $2a$ is even and $\triangleleft_1 = \, <$, $\vb_{ij}^l(2a) = 0$.
	\qed
\end{proof}

\invvolnotdisjoint*
\begin{proof}
	Assume first that there exists a vol-minimal witness with finite volume.	Suppose that the sets of vol- and inv-minimal witnesses were disjoint.
	Then for each vol-minimal witness $\T_1$ there must exist another witness $\T_2$ such that $\T_2 <_{\inv} \T_1$, as otherwise $\T_1$ would be inv-minimal.
	By definition of $\leq_{\inv}$ it follows that $\vol(\T_2) \leq \vol(\T_1)$ and as $\T_1$ is vol-minimal, we get $\vol(\T_2) = \vol(\T_1)$.
	Iterating this argument yields an infinitely descending chain of finite-volume subsystems that are all strictly smaller in the $\leq_{\inv}$ order.
	But this cannot exist, as the relation $<_{\inv}$ over finite-volume subsystems of $\T$ is well-founded.
	
	Now suppose that a vol-minimal witness for $\prb^*_\T(\lozenge\goal)\geq \lambda$ has infinite volume. Then, trivially, any witness for $\prb^*_\T(\lozenge\goal)\geq \lambda$ is vol-minimal since they all have infinite volume. In particular, every inv-minimal witness is also vol-minimal.
	\qed
\end{proof}

\dbmhashp*
\begin{proof}
  From the proof of~\cite[Theorem 5.1.4]{GritzmannK1994} it follows that volume computation is \hashP-hard already for polytopes of the form
  \[ \P^I = \{x \in [0,1]^n \mid \forall (i,j) \in I. \; x(i) \leq x(j) \}\]
  for a given $I \subseteq \{1,\ldots,n\}^2$.
  On the other hand, such a polytope can be defined using a DBM over clocks $\C = \{c_0,\ldots,c_n\}$ as follows:
  \[M^I_{ij} = \begin{cases}
    (1, \leq) & \text{ if } i\geq 1, j = 0 \\
    (0, \leq) & \text{ if } i= 0, j \geq 0 \\
    (0, \leq) &  \text{ if } (i,j) \in I \\
    (1, \leq) & \text{ otherwise}
  \end{cases} \]
  The first two cases represent the constraint $0 \leq c_i \leq 1$ for all clocks.
  The third case formalizes that $c_i - c_j \leq 0$ should hold whenever $(i,j) \in I$. Given that $0\leq c_i\leq 1$, the fourth condition does not impose any further restriction on the polytope. Then $\P^I$ equals $\Val(M^I)$ considered as a subset of $\mathbb{R}^{\C\setminus\{c_0\}} \cong \mathbb{R}^n$, and hence $\vol(\P^I) = \vol(\Val(M^I))$.
  \qed
 \end{proof}

\leqvolpp*
\begin{proof}
		As the problem of computing $\vol(\Val(M))$ is \hashP-hard by~\Cref{prop:dbmvolhashp}, it follows that the corresponding threshold problem $\vol(\Val(M)) \geq k$ for a given $k \in \mathbb{Q}$ is \PP-hard under polynomial-time Turing reductions.
		By the proof of~\Cref{prop:dbmvolhashp} it follows that it is hard already for DBM that have $1$ as an upper bound for each variable, and only use $\leq$ comparisons.
		We show that computing the volume-threshold problem for such DBM can be reduced to deciding whether $\T_1 \leq_{\vol} \T_2$ holds given a PTA $\T$ and two subsystems $\T_1,\T_2$ of $\T$.
		
		Let $M$ be such a DBM over $n$ clocks.
		We let $\T$ be the PTA that has two locations $l_1, l_2$ with invariants $M_1$ and $M_2$, respectively, defined as follows.
		$M_1$ inherits all its entries from $M$, apart from the upper bounds (that is, comparisons with the zero-clock) which are set to $n!$.
		As $n! = \mathcal{O}(2^{n \, \log n})$, we can express $n!$ in $\poly(n)$ bits.
		We have: $\vol(\Val(M_1)) = n!^n \cdot \vol(\Val(M))$.
		Hence $\vol(\Val(M)) \geq k$ is equivalent to $\vol(\Val(M_1)) \geq k \cdot n!^{n}$.
		
		As $\vol(\Val(M))$ is a multiple of $1/n!$ \footnote{This follows from the fact that $n!$ different full-dimensional ``regions'' in the 1-cube with the same size can be distinguished by a DBM. They correspond to the possible relative values of each pair of clocks, which in turn corresponds to the possible permutations of $1,\ldots,n$. } we can assume that so is $k$ (or we round up to the nearest rational with this property). 
		We let $M_2$ be the DBM that describes a row of $k \cdot n!^n$ (which is an integer) 1-cubes in $n$ dimensions.
		This is achieved by letting all variables have upper bound $1$ apart from a single variable with upper bound $k \cdot n!^{n}$ (note that $k \cdot n!^{n} = \mathcal{O}(k \cdot (2^{n \log n})^{n}) = \mathcal{O}(k \cdot (2^{n^2 \cdot \log n}))$ and hence expressible with $\poly(n) + \log(k)$ many bits).
		We have $\vol(\Val(M_2)) = k \cdot n!^{n}$.
		
		Now let $\T_1$ be the subsystem that includes only location $l_1$, and $\T_2$ be the subsystem that includes only location $l_2$.
		Then we have $\vol(\Val(M)) \geq k$ iff $\T_2 \leq_{\vol} \T_1$, which completes the reduction of the threshold problem for the volume of valuation sets of DBMs to deciding $\leq_{\vol}$.
		\qed
	\end{proof}

\end{document}

%% file: tikzstyle.tex
\usepackage{tikz}
\usetikzlibrary{positioning,arrows,automata,calc,shapes}

\definecolor{darkgreen}{rgb}{0.0, 0.5, 0.0}
\definecolor{darkred}{rgb}{0.9, 0.0, 0.0}
\definecolor{darkblue}{rgb}{0.0, 0.0, 0.9}

\tikzset{
	state/.style={rounded rectangle,draw=black,inner sep=1.5mm,minimum width=9mm,minimum height=5mm},
	rstate/.style={rectangle,draw=black,inner sep=1ex,minimum width=9mm,minimum height=5mm},
	istate/.style={minimum width=5mm, minimum height=6mm},
	bullet/.style={circle,draw=black,fill=black,inner sep=0cm, minimum size=0.7mm},
	initial text={},
	every initial by arrow/.style={->,>=stealth'},
	ptran/.style={rounded corners, ->,>=stealth',auto},
	ntran/.style={rounded corners, -,auto}
}

%% file: pta_example.tex
\begin{figure}[tbp]
	\captionsetup{labelfont={up}}
	\begin{subfigure}[b]{0.5\columnwidth}
		\centering
		\scalebox{1.0}{
	\begin{tikzpicture}[->,>=stealth',shorten >=1pt,auto,node distance=0.5cm, semithick, x=18mm,y=15mm,font=\scriptsize]
	  \node[state] (l0) {$l_0,x=0$};
	  \node[state] (l1) [below left = 0.8cm and 0.5cm of l0]  {$l_1, x \leq 2$};
	  \node[state] (l2) [below right = 0.8cm and 0.5cm of l0]  {$l_2, y \leq 2$};
	  \node[state] (l3) [below left = 2.2cm and 0.5cm of l0]  {$l_3,x \leq 1$};
	  \node[state] (goal) [below left = 4.2cm and 0.6cm of l0]  {$\:\goal\:$};
	  \node[state] (fail) [below right = 4.25cm and 0.6cm of l0]  {$\:\:\fail\:\:$};
	   
	    \node[bullet] (l0bullet) [below = 0.5cm of l0] {};
	    \node[bullet] (l1bullet) [left = 0.7cm of l1] {};
	    \node[bullet] (l2bullet) [below left = 2.2cm and 0.1cm of l2] {};
	    \node[bullet] (l3bullet) [below right = 0.6cm and 0.2cm of l3] {};
	    \node [below left = -0.3cm and 1.4cm of l0] {$x:=0$};
	    \node [below left = 0.75cm and 1.7cm of l0] {$\alpha$};
	    \node [below left = 1.36cm and 0.9cm of l0] {$\beta$};
	    
	     \draw (l1) -- node[right=0.0,pos=.3]{$x\leq 1$} (l3);
	    
	    { \color{black}
	    	\draw[-] (l0) to (l0bullet);
	    	\draw (l0bullet) edge[bend left] node[right, above = 0., pos=0.7]  {$\frac{2}{5}$} (l1);
	    	\draw (l0bullet) edge[bend right] node[left, above = 0., pos=0.7]  {$\frac{3}{5}$} (l2);
	    }
	
		{ \color{black}
			\draw[-] (l1) -- node[below=0.02,pos=.5]{$x\geq 1$} (l1bullet);
			\draw (l1bullet) edge[bend left] node[right, above = 0.0, pos=0.8]  {$\frac{1}{2}$} (l0);
			\draw (l1bullet) edge[bend right] node[left, above = 0.0, pos=0.9]  {$\frac{1}{2}$} (goal);
		}
		
		{ \color{black}
			\draw[-] (l2) -- node[right=0.02,pos=0.15]{$y\geq 1$} (l2bullet);
			\draw (l2bullet) edge[bend left] node[right, above = 0.0, pos=0.4]  {$\frac{2}{5}$} (goal);
			\draw (l2bullet) edge[bend right] node[left, right = 0.0, pos=0.5]  {$\frac{3}{5}$} (fail);
		}
	
		 { \color{black}
			\draw[-] (l3) -- node[right=0.03,pos=.3]{$x\geq 1$} (l3bullet);
			\draw (l3bullet) edge[bend left] node[right, left = 0.0, pos=0.5]  {$\frac{3}{4}$} (goal);
			\draw (l3bullet) edge[bend right] node[left, above = 0.0, pos=0.4]  {$\frac{1}{4}$} (fail);
		}
	\end{tikzpicture}}
		\subcaption{}
	\label{fig:example_pta}
\end{subfigure}
	\begin{subfigure}[b]{0.5\columnwidth}
		\centering
		\scalebox{1.0}{
				\begin{tikzpicture}[->,>=stealth',shorten >=1pt,auto,node distance=0.5cm, semithick, x=18mm,y=15mm,font=\scriptsize]
					\node[state] (l0) {$l_0,x=0$};
					\node[state] (l1) [below left = 0.8cm and 0.5cm of l0]  {$l_1, x \leq 2$};
					\node[state] (l2) [below right = 0.8cm and 0.5cm of l0]  {$l_2, y \leq 2$};
					\node[state] (l3) [below left = 2.2cm and 0.1cm of l0]  {$l_3,x \leq 1$, {\color{red} $y\leq 1$}};
					\node[state] (goal) [below left = 4.2cm and 0.6cm of l0]  {$\:\goal\:$};
					\node[state] (fail) [below right = 4.25cm and 0.6cm of l0]  {$\:\:\fail\:\:$};
					
					\node[bullet] (l0bullet) [below = 0.5cm of l0] {};
					\node[bullet] (l1bullet) [left = 0.7cm of l1] {};
					\node[bullet] (l2bullet) [below left = 2.2cm and 0.1cm of l2] {};
					\node [below left = -0.3cm and 1.4cm of l0] {$x:=0$};
					
					\node [below left = 0.75cm and 1.7cm of l0] {$\alpha$};
					\node [below left = 1.36cm and 0.9cm of l0] {$\beta$};
					\draw (l1) -- node[right=0.0,pos=.3]{$x\leq 1$} (l3);
					{\color{red}
					\draw (l3) to[bend right] node[above=0.1,right = 0.03, pos=.08]{{\color{black}$x\geq 1$}} (fail);
						    
				}
					
					{ \color{black}
						\draw[-] (l0) to (l0bullet);
						\draw (l0bullet) edge[bend left] node[right, above = 0., pos=0.7]  {$\frac{2}{5}$} (l1);
						\draw (l0bullet) edge[bend right] node[left, above = 0., pos=0.7]  {$\frac{3}{5}$} (l2);
					}
					
					{ \color{black}
						\draw[-] (l1) -- node[below=0.02,pos=.5]{{\color{red}$x\geq 2$}} (l1bullet);
						\draw (l1bullet) edge[bend left] node[right, above = 0.0, pos=0.8]  {$\frac{1}{2}$} (l0);
						\draw (l1bullet) edge[bend right] node[left, above = 0.0, pos=0.9]  {$\frac{1}{2}$} (goal);
					}
					
					{ \color{black}
						\draw[-] (l2) -- node[right=0.02,pos=0.15]{{\color{red}$y=2$}} (l2bullet);
						\draw (l2bullet) edge[bend left] node[right, above = 0.0, pos=0.35]  {$\frac{2}{5}$} (goal);
						\draw (l2bullet) edge[bend right] node[left, right = 0.0, pos=0.5]  {$\frac{3}{5}$} (fail);
					}
					
					{ \color{black}
					}
				\end{tikzpicture}
	}
		\subcaption{}
		\label{fig:example_subsystem}
	\end{subfigure}
\caption{A pointed PTA (left) and a weak subsystem therein (right).}
\end{figure}
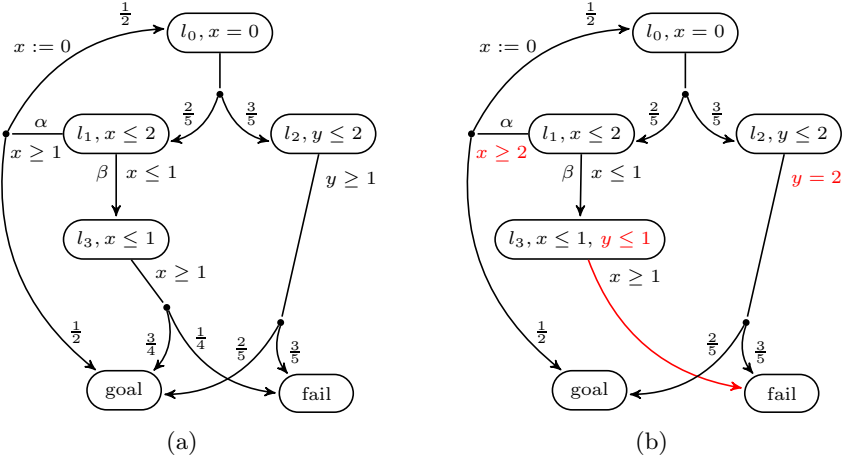

%% file: line_vertical.tex
   \begin{tikzpicture}[baseline=(O.base)] {
	\node (O)  {};}
      \begin{axis}[
          width=2.1cm, height=2.1cm,
          xmin=0, xmax=1, xmajorticks=false,
          ymin=0, ymax=1, ymajorticks=false
         ]
        \path[name path=xaxis] (axis cs:0,0) -- (axis cs:1,0);
        \path[name path=yaxis] (axis cs:0,0) -- (axis cs:0,1);

        \addplot[name path=line1,ultra thick,color=blue,mark=*, mark options={scale=\pointsize}]
        coordinates
        {(0,0) (0,1)};
      \end{axis}
    \end{tikzpicture}

%% file: parallelogram.tex
   \begin{tikzpicture}[baseline=(O.base)] {
	\node (O)  {};}
		\begin{axis}[
		width=2.1cm, height=2.65cm,
		xmin=0, xmax=1, xmajorticks=false,
		ymin=0, ymax=2, ymajorticks=false
		]
		\path[name path=xaxis] (axis cs:0,0) -- (axis cs:1,0);
		\path[name path=yaxis] (axis cs:0,0) -- (axis cs:0,2);

        \addplot[name path=line1, ultra thick,color=blue,mark=*, mark options={scale=\pointsize}]
        coordinates
        {(0,0) (1,1)};

		\addplot[name path=line2,ultra thick,color=blue,mark=*, mark options={scale=\pointsize}]
		coordinates
		{(0,0) (0,1)};
		
		\addplot[name path=line3,ultra thick,color=blue,mark=*, mark options={scale=\pointsize}]
		coordinates
		{(0,1) (1,2)};
		
		\addplot[name path=line4,ultra thick,color=blue,mark=*, mark options={scale=\pointsize}]
		coordinates
		{(1,1) (1,2)};
		
		\addplot[thick,
		color=blue,
		fill=blue, 
		fill opacity=0.15]
		fill between[
		of = line1 and line3
		];
		
		\addplot[name path=line3, thin,color=gray,mark=none]
		coordinates
		{(0,1) (1,1)};
      \end{axis}
    \end{tikzpicture}

%% file: point_00.tex
   \begin{tikzpicture}[baseline=(O.base)] {
	\node (O)  {};}
      \begin{axis}[
          width=2.1cm, height=2.1cm,
          xmin=0, xmax=1, xmajorticks=false, 
          ymin=0, ymax=1, ymajorticks=false
         ]
        \path[name path=xaxis] (axis cs:0,0) -- (axis cs:1,0);
        \path[name path=yaxis] (axis cs:0,0) -- (axis cs:0,1);

        \addplot[name path=line1,ultra thick,color=blue,mark=*, mark options={scale=\pointsize}]
        coordinates
        {(0,0) (0,0)};
      \end{axis}
    \end{tikzpicture}

%% file: line.tex
   \begin{tikzpicture}[baseline=(O.base)] {
	\node (O)  {};}
      \begin{axis}[
          width=2.1cm, height=2.1cm,
          xmin=0, xmax=1, xmajorticks=false,
          ymin=0, ymax=1, ymajorticks=false
         ]
        \path[name path=xaxis] (axis cs:0,0) -- (axis cs:1,0);
        \path[name path=yaxis] (axis cs:0,0) -- (axis cs:0,1);

        \addplot[name path=line1,ultra thick,color=blue, mark=*, mark options={scale=\pointsize}]
        coordinates
        {(0,0) (1,1)};
      \end{axis}
    \end{tikzpicture}

%% file: point_11.tex
   \begin{tikzpicture}[baseline=(O.base)] {
	\node (O)  {};}
      \begin{axis}[
          width=2.1cm, height=2.1cm,
          xmin=0, xmax=1, xmajorticks=false,
          ymin=0, ymax=1, ymajorticks=false
          ]
        \path[name path=xaxis] (axis cs:0,0) -- (axis cs:1,0);
        \path[name path=yaxis] (axis cs:0,0) -- (axis cs:0,1);

        \addplot[name path=line1,ultra thick,color=blue,mark=*, mark options={scale=\pointsize}]
        coordinates
        {(1,1) (1,1)};
      \end{axis}
    \end{tikzpicture}

%% file: long_line.tex
   \begin{tikzpicture}[baseline=(O.base)] {
	\node (O)  {};}
      \begin{axis}[
          width=2.65cm, height=2.65cm,
          xmin=0, xmax=2, xmajorticks=false,
          ymin=0, ymax=2, ymajorticks=false
          ]
        \path[name path=xaxis] (axis cs:0,0) -- (axis cs:2,0);
        \path[name path=yaxis] (axis cs:0,0) -- (axis cs:0,2);
        
        \addplot[name path=line1,ultra thick,color=blue,mark=*, mark options={scale=\pointsize}]
        coordinates
        {(0,0) (2,2)};
        
        \addplot[name path=line2, thin,color=gray,mark=none]
        coordinates
        {(1,0) (1,2)};
        
        \addplot[name path=line3, thin,color=gray,mark=none]
        coordinates
        {(0,1) (2,1)};
      \end{axis}
    \end{tikzpicture}

%% file: long_line_vertical.tex
   \begin{tikzpicture}[baseline=(O.base)] {
	\node (O)  {};}
      \begin{axis}[
          width=2.1cm, height=2.65cm,
          xmin=0, xmax=1, xmajorticks=false,
          ymin=0, ymax=2, ymajorticks=false
         ]
        \path[name path=xaxis] (axis cs:0,0) -- (axis cs:1,0);
        \path[name path=yaxis] (axis cs:0,0) -- (axis cs:0,2);

        \addplot[name path=line1,ultra thick,color=blue,mark=*, mark options={scale=\pointsize}]
        coordinates
        {(0,0) (0,2)};
        
        \addplot[name path=line3, thin,color=gray,mark=none]
        coordinates
        {(0,1) (1,1)};
      \end{axis}
    \end{tikzpicture}

%% file: long_parallelogram.tex
   \begin{tikzpicture}[baseline=(O.base)] {
	\node (O)  {};}
		\begin{axis}[
		width=2.65cm, height=3.75cm,
		xmin=0, xmax=2, xmajorticks=false,
		ymin=0, ymax=4, ymajorticks=false
		]
		\path[name path=xaxis] (axis cs:0,0) -- (axis cs:2,0);
		\path[name path=yaxis] (axis cs:0,0) -- (axis cs:0,4);

        \addplot[name path=line1, ultra thick,color=blue,mark=*, mark options={scale=\pointsize}]
        coordinates
        {(0,0) (2,2)};

		\addplot[name path=line2,ultra thick,color=blue,mark=*, mark options={scale=\pointsize}]
		coordinates
		{(0,0) (0,2)};
		
		\addplot[name path=line3,ultra thick,color=blue,mark=*, mark options={scale=\pointsize}]
		coordinates
		{(0,2) (2,4)};
		
		\addplot[name path=line4,ultra thick,color=blue,mark=*, mark options={scale=\pointsize}]
		coordinates
		{(2,2) (2,4)};
		
		\addplot[thick,
		color=blue,
		fill=blue, 
		fill opacity=0.15]
		fill between[
		of = line1 and line3
		];
		
		\addplot[name path=line3, thin,color=gray,mark=none]
		coordinates
		{(1,0) (1,4)};
		
		\addplot[name path=line3, thin,color=gray,mark=none]
		coordinates
		{(0,1) (2,1)};
		
		\addplot[name path=line3, thin,color=gray,mark=none]
		coordinates
		{(0,2) (2,2)};
		
		\addplot[name path=line3, thin,color=gray,mark=none]
		coordinates
		{(0,3) (2,3)};

      \end{axis}
    \end{tikzpicture}

%% file: vertical_parallelogram.tex
   \begin{tikzpicture}[baseline=(O.base)] {
	\node (O)  {};}
		\begin{axis}[
		width=2.1cm, height=3.75cm,
		xmin=0, xmax=1, xmajorticks=false,
		ymin=0, ymax=4, ymajorticks=false
		]
		\path[name path=xaxis] (axis cs:0,0) -- (axis cs:1,0);
		\path[name path=yaxis] (axis cs:0,0) -- (axis cs:0,4);

        \addplot[name path=line1, ultra thick,color=blue,mark=*, mark options={scale=\pointsize}]
        coordinates
        {(0,0) (1,1)};

		\addplot[name path=line2,ultra thick,color=blue,mark=*, mark options={scale=\pointsize}]
		coordinates
		{(0,0) (0,2)};
		
		\addplot[name path=line3,ultra thick,color=blue,mark=*, mark options={scale=\pointsize}]
		coordinates
		{(0,2) (1,3)};
		
		\addplot[name path=line4,ultra thick,color=blue,mark=*, mark options={scale=\pointsize}]
		coordinates
		{(1,1) (1,3)};
		
		\addplot[thick,
		color=blue,
		fill=blue, 
		fill opacity=0.15]
		fill between[
		of = line1 and line3
		];
		
		\addplot[name path=line3, thin,color=gray,mark=none]
		coordinates
		{(1,0) (1,4)};
		
		\addplot[name path=line3, thin,color=gray,mark=none]
		coordinates
		{(0,1) (1,1)};
		
		\addplot[name path=line3, thin,color=gray,mark=none]
		coordinates
		{(0,2) (1,2)};
		
		\addplot[name path=line3, thin,color=gray,mark=none]
		coordinates
		{(0,3) (1,3)};

      \end{axis}
    \end{tikzpicture}

%% file: main.bbl
\begin{thebibliography}{10}
\providecommand{\url}[1]{\texttt{#1}}
\providecommand{\urlprefix}{URL }
\providecommand{\doi}[1]{https://doi.org/#1}

\bibitem{AlurCD93}
Alur, R., Courcoubetis, C., Dill, D.: Model-checking in dense real-time.
  Information and Computation  \textbf{104}(1),  2 -- 34 (1993).
  \doi{10.1006/inco.1993.1024}

\bibitem{AlurD1994}
Alur, R., Dill, D.L.: A theory of timed automata. Theoretical Computer Science
  \textbf{126}(2),  183--235 (Apr 1994). \doi{10.1016/0304-3975(94)90010-8}

\bibitem{AndresDR08}
Andr{\'{e}}s, M.E., D'Argenio, P.R., van Rossum, P.: Significant diagnostic
  counterexamples in probabilistic model checking. In: 4th International Haifa
  Verification Conference, {HVC} (2008). \doi{10.1007/978-3-642-01702-5\_15}

\bibitem{DBLP:books/daglib/0023084}
Arora, S., Barak, B.: Computational Complexity - {A} Modern Approach. Cambridge
  University Press (2009)

\bibitem{BaierK2008}
Baier, C., Katoen, J.P.: Principles of Model Checking (Representation and Mind
  Series). The MIT Press, Cambridge, MA (2008)

\bibitem{Beauquier03}
Beauquier, D.: On probabilistic timed automata. Theor. Comput. Sci.
  \textbf{292}(1),  65–84 (2003). \doi{10.1016/S0304-3975(01)00215-8}

\bibitem{Behrmann_etal06}
Behrmann, G., David, A., Larsen, K.G., H{\aa}kansson, J., Petterson, P., Yi,
  W., Hendriks, M.: Uppaal 4.0. In: Quantitative Evaluation of Systems. QEST
  (2006). \doi{10.1109/QEST.2006.59}

\bibitem{BengtssonW04}
Bengtsson, J., Yi, W.: Timed Automata: Semantics, Algorithms and Tools, pp.
  87--124. Springer Berlin Heidelberg (2004).
  \doi{10.1007/978-3-540-27755-2\_3}

\bibitem{BerendsenJK06}
Berendsen, J., Jansen, D.N., Katoen, J.: Probably on time and within budget: On
  reachability in priced probabilistic timed automata. In: Quantitative
  Evaluation of Systems {QEST} (2006). \doi{10.1109/QEST.2006.43}

\bibitem{CeskaHJK19}
Ceska, M., Hensel, C., Junges, S., Katoen, J.: Counterexample-driven synthesis
  for probabilistic program sketches. In: Formal Methods - The Next 30 Years -
  Third World Congress, {FM} 2019 (2019). \doi{10.1007/978-3-030-30942-8\_8}

\bibitem{ChenHK08}
{Chen}, T., {Han}, T., {Katoen}, J.: Time-abstracting bisimulation for
  probabilistic timed automata. In: International Symposium on Theoretical
  Aspects of Software Engineering. pp. 177--184 (2008).
  \doi{10.1109/TASE.2008.29}

\bibitem{DierksKL2007}
Dierks, H., Kupferschmid, S., Larsen, K.G.: Automatic {{Abstraction
  Refinement}} for {{Timed Automata}}. In: Formal {{Modeling}} and {{Analysis}}
  of {{Timed Systems}}. {Springer} (2007). \doi{10.1007/978-3-540-75454-1\_10}

\bibitem{Dill1990}
Dill, D.L.: Timing assumptions and verification of finite-state concurrent
  systems. In: Automatic {{Verification Methods}} for {{Finite State Systems}}.
  Lecture {{Notes}} in {{Computer Science}}, {Springer} (1990).
  \doi{10.1007/3-540-52148-8\_17}

\bibitem{FunkeJB20}
Funke, F., Jantsch, S., Baier, C.: Farkas certificates and minimal witnesses
  for probabilistic reachability constraints. In: Tools and Algorithms for the
  Construction and Analysis of Systems (TACAS). Springer (2020).
  \doi{10.1007/978-3-030-45190-5\_18}

\bibitem{GritzmannK1994}
Gritzmann, P., Klee, V.: On the {{Complexity}} of {{Some Basic Problems}} in
  {{Computational Convexity}}. In: Polytopes: {{Abstract}}, {{Convex}} and
  {{Computational}}. {Springer Netherlands} (1994).
  \doi{10.1007/978-94-011-0924-6\_17}

\bibitem{HermannsWZ2008}
Hermanns, H., Wachter, B., Zhang, L.: Probabilistic {{CEGAR}}. In: Computer
  {{Aided Verification}}. pp. 162--175. Lecture {{Notes}} in {{Computer
  Science}}, {Springer}, {Berlin, Heidelberg} (2008).
  \doi{10.1007/978-3-540-70545-1\_16}

\bibitem{JansenAKWKB11}
Jansen, N., {\'{A}}brah{\'{a}}m, E., Katelaan, J., Wimmer, R., Katoen, J.,
  Becker, B.: Hierarchical counterexamples for discrete-time {M}arkov chains.
  In: Automated Technology for Verification and Analysis, 9th International
  Symposium, {ATVA} (2011). \doi{10.1007/978-3-642-24372-1\_33}

\bibitem{JansenWAZKBS14}
Jansen, N., Wimmer, R., {\'{A}}brah{\'{a}}m, E., Zajzon, B., Katoen, J.,
  Becker, B., Schuster, J.: Symbolic counterexample generation for large
  discrete-time {M}arkov chains. Science of Computer Programming  \textbf{91},
  90--114 (2014). \doi{10.1016/j.scico.2014.02.001}

\bibitem{JurdzinskiKNT09}
Jurdzi{\'{n}}ski, M., Kwiatkowska, M., Norman, G., Trivedi, A.:
  Concavely-priced probabilistic timed automata. In: Concurrency Theory
  (CONCUR). Springer (2009). \doi{10.1007/978-3-642-04081-8\_28}

\bibitem{JurdzinskiLS07}
Jurdzi{\'{n}}ski, M., Laroussinie, F., Sproston, J.: Model checking
  probabilistic timed automata with one or two clocks. In: Tools and Algorithms
  for the Constr. and Analysis of Systems. Springer (2007).
  \doi{10.1007/978-3-540-71209-1\_15}

\bibitem{KolblLW2019}
K{\"o}lbl, M., Leue, S., Wies, T.: Clock {{Bound Repair}} for {{Timed
  Systems}}. In: Computer {{Aided Verification}}. pp. 79--96. Lecture {{Notes}}
  in {{Computer Science}}, {Springer Intern. Publishing}, {Cham} (2019).
  \doi{10.1007/978-3-030-25540-4\_5}

\bibitem{KwiatkowskaNSS02}
Kwiatkowska, M., Norman, G., Segala, R., Sproston, J.: Automatic verification
  of real-time systems with discrete probability distributions. Theoretical
  Computer Science  \textbf{282}(1),  101 -- 150 (2002).
  \doi{10.1016/S0304-3975(01)00046-9}

\bibitem{KwiatkowskaNS03}
Kwiatkowska, M., Norman, G., Sproston, J.: {Probabilistic Model Checking of
  Deadline Properties in the IEEE 1394 FireWire Root Contention Protocol}.
  Formal Aspects of Comput.  \textbf{14}(3),  295–318 (2003).
  \doi{10.1007/s001650300007}

\bibitem{KwiatkowskaNPS07}
Kwiatkowska, M.Z., Norman, G., Parker, D., Sproston, J.: Performance analysis
  of probabilistic timed automata using digital clocks. Form Method Syst Des
  \textbf{29},  33–78 (2006). \doi{10.1007/s10703-006-0005-2}

\bibitem{KwiatkowskaNSW07}
Kwiatkowska, M.Z., Norman, G., Sproston, J., Wang, F.: Symbolic model checking
  for probabilistic timed automata. Information and Computation
  \textbf{205}(7),  1027--1077 (2007). \doi{10.1016/j.ic.2007.01.004}

\bibitem{LaroussinieS07}
Laroussinie, F., Sproston, J.: State explosion in almost-sure probabilistic
  reachability. Inf. Process. Lett.  \textbf{102}(6),  236–241 (2007).
  \doi{10.1016/j.ipl.2007.01.003}

\bibitem{NormanPS13}
Norman, G., Parker, D., Sproston, J.: {Model checking for probabilistic timed
  automata}. Formal Methods in System Design  \textbf{43},  164–190 (2013).
  \doi{10.1007/s10703-012-0177-x}

\bibitem{OzpeynirciK2010}
{\"O}zpeynirci, {\"O}., K{\"o}ksalan, M.: An exact algorithm for finding
  extreme supported nondominated points of multiobjective mixed integer
  programs. Management Science  \textbf{56}(12),  2302--2315 (2010).
  \doi{10.1287/mnsc.1100.1248}

\bibitem{PetterssonO2019}
Pettersson, W., Ozlen, M.: Multi-{{Objective Mixed Integer Programming}}: {{An
  Objective Space Algorithm}}. AIP Conference Proceedings  \textbf{2070}(1),
  020039 (2019). \doi{10.1063/1.5090006}

\bibitem{Sproston11}
Sproston, J.: {Discrete-Time Verification and Control for Probabilistic
  Rectangular Hybrid Automata}. In: Eigth International Conference on
  Quantitative Evaluation of Systems, QEST 2011. pp. 79 -- 88 (2011).
  \doi{10.1109/QEST.2011.18}

\bibitem{Tripakis98}
Tripakis, S.: {L'analyse formelle des syst{\`e}mes temporis{\`e}s en pratique}.
  {Ph.D. thesis, Universit{\'e} Joseph Fourier}  (1998)

\bibitem{WimmerJAK15}
Wimmer, R., Jansen, N., {\'{A}}brah{\'a}m, Erika~Katoen, J.P.: {High-level
  Counterexamples for Probabilistic Automata}. {Logical Methods in Computer
  Science}  \textbf{11}(1) (2015). \doi{10.2168/LMCS-11(1:15)2015}

\bibitem{WimmerJAKB14}
Wimmer, R., Jansen, N., {\'{A}}brah{\'{a}}m, E., Katoen, J., Becker, B.:
  Minimal counterexamples for linear-time probabilistic verification.
  Theoretical Computer Science  \textbf{549},  61--100 (2014).
  \doi{10.1016/j.tcs.2014.06.020}

\bibitem{WimmerM20}
Wimmer, S., Mutius, J.v.: Verified certification of reachability checking for
  timed automata. In: Tools and Algorithms for the Construction and Analysis of
  Systems (TACAS). Springer (2020). \doi{10.1007/978-3-030-45190-5\_24}

\end{thebibliography}
